\documentclass[journal]{IEEEtran}
\usepackage{amsmath}
\usepackage{amsfonts}
\usepackage{float}
\usepackage{graphicx}
\usepackage{amsthm}
\usepackage{amssymb}
\newtheorem{assumption}{Assumption}
\newtheorem{problem}{Problem}
\newtheorem{theorem}{Theorem}
\newtheorem{algorithm}{Algorithm}
\newtheorem{lemma}{Lemma}
\newtheorem{remark}{Remark}
\usepackage{accents}
\usepackage{color}
\usepackage{algorithm}
\usepackage[noend]{algpseudocode}
\def\BState{\State\hskip-\ALG@thistlm}
\usepackage{supertabular,booktabs}
\usepackage{cite}
\usepackage{multicol,lipsum}
\usepackage{stfloats}
\usepackage{arydshln}
\usepackage{mathtools}
\usepackage{graphicx}
\usepackage{subcaption}

\newcommand{\norm}[1]{\left\lVert#1\right\rVert}
\newcommand{\RNum}[1]{\uppercase\expandafter{\romannumeral #1\relax}}
\graphicspath{ {images/} }

\usepackage{tikz}
\usetikzlibrary{shapes,arrows}
\usetikzlibrary{positioning}	
	
\newtheoremstyle{theoremdd}
{\topsep}
{\topsep}
{\itshape}
{0pt}
{\bfseries}
{:}
{ }
{\thmname{#1}\thmnumber{ #2} \thmnote{(#3)}}

\theoremstyle{theoremdd}
\newtheorem{definition}{Definition}

\usepackage{chngcntr}
\usepackage{apptools}
\AtAppendix{\counterwithin{assumption}{section}}
\AtAppendix{\counterwithin{definition}{section}}
\AtAppendix{\counterwithin{theorem}{section}}
\ifCLASSINFOpdf
\else
\fi
\begin{document}
%
\title{Learning-Based Adaptive Optimal Control of Linear Time-Delay Systems: A Policy Iteration Approach}
%
%
%

\author{Leilei Cui, 
         Bo Pang %
         and Zhong-Ping Jiang,~\IEEEmembership{Fellow,~IEEE}
\thanks{*This work is supported partly by the National Science Foundation under Grant EPCN-1903781.}
\thanks{L. Cui, B. Pang, and Z.~P.~Jiang are with the Control and Networks Lab, Department of Electrical and Computer Engineering, Tandon School of Engineering, New York University, Brooklyn, NY 11201, USA (e-mail: l.cui@nyu.edu; bo.pang@nyu.edu; zjiang@nyu.edu).}
}

%
%

\markboth{Manuscript}%
{Shell \MakeLowercase{\textit{et al.}}: Bare Demo of IEEEtran.cls for IEEE Journals}
%



\maketitle

\begin{abstract}
This paper studies the adaptive optimal control problem for a class of linear time-delay systems described by delay differential equations (DDEs). A crucial strategy is to take advantage of recent developments in reinforcement learning and adaptive dynamic programming and develop novel methods to learn adaptive optimal controllers from finite samples of input and state data. In this paper, the data-driven policy iteration (PI) is proposed to solve the infinite-dimensional algebraic Riccati equation (ARE) iteratively in the absence of exact model knowledge.  Interestingly, the proposed recursive PI algorithm is new in the present context of continuous-time time-delay systems, even when the model knowledge is assumed known. The efficacy of the proposed learning-based control methods is validated by means of practical applications arising from metal cutting and autonomous driving.
\end{abstract}


%
\IEEEpeerreviewmaketitle

\section{Introduction}
Time-delay systems have attracted considerable attention, in large part because they are ubiquitous in many branches of science and engineering; see the books \cite{Book_Kolmanovskii,Book_Hale, Book_Karafyllis} for many references and examples. Recently, many theoretical results are developed for time-delay systems, such as input-to-state stability of time-delay systems\cite{Pepe2006,Yeganefar2008}, robust $H_\infty$ control of distributed delay systems \cite{Xie2001}, stabilization of large-scale stochastic systems with time delays \cite{Xie2000}, and stability analysis of systems with time-varying delay \cite{He2007, Gao2007, Lam2007}. Examples of time-delay systems are in transportation \cite{Ge2017, Huang2021}, biological motor control \cite{Gabor2009}, network-based control \cite{GAO2008}, communication networks \cite{Hollot2002}, multi-agent systems \cite{Liu2011,Tang2015}, and heating systems \cite{Vyhlidal2003}. It is thus not surprising that the optimal control problem of time-delay systems has been a fundamentally important, yet challenging, research topic in control theory for several decades. For instance, for a linear time-varying system with time-delay,  Eller \textit{et al.} \cite{Eller1969} proposed a first solution to the finite-horizon linear quadratic (LQ) optimal control problem. Ross \textit{et al.} \cite{Ross1969optimal, Ross1971TAC} presented a linear control law to solve the infinite horizon LQ optimal control problem for a class of linear time-delay systems. In these papers, in order to calculate the desired optimal control law, the certain infinite-dimensional Riccati equations have to be solved. For this problem, many algorithms have been developed \cite{Gibson1983, Banks1984, Burns2008}. However, an accurate model of the time-delay system is required for these algorithms, and in reality, it is difficult to derive an exact model due to the complexity of the system and the inevitable system uncertainties. Therefore, developing a model-free optimal control approach for time-delay systems is a timely research topic of both theoretical importance and practical relevance that requires further investigation. Recent progresses and successes in reinforcement learning (RL) provide an opportunity to advance the state of the art in the area of adaptive optimal control of time-delay systems.

RL is an important branch of machine learning and is aimed at 
maximizing (or minimizing) the cumulative reward (or cost) through continuous agent-environment interactions and exploration of policies. Under the RL framework, an agent is capable of learning an optimal policy (controller) when the environment is even unknown. A notable example of recent success is the development of advanced deep RL algorithms in the game of GO and many video games \cite{Silver2013,Silver2016}, that have achieved human-level intelligence leading to widespread attention from academia and industry.
Despite these successes, traditional RL has some fundamental limitations. For example, it often assumes that the environment is depicted by Markov decision processes and discrete-time
systems with finite or countable state-action space. Often, the stability aspect of the learned controller by RL is not guaranteed. For many systems described by differential equations, such as autonomous vehicles and quadrupedal robots, the state and action spaces are infinite and the stability of the controller generated by an RL algorithm is innegligible. Therefore, for these safety-critical engineering systems, conventional RL is not directly applicable to learning stable optimal controllers from data, which has motivated the development of adaptive dynamic programming \cite{Book_Jiang,Book_Lewis}. In contrast with conventional RL, the purpose of continuous-time ADP is aimed at addressing decision making problems for dynamical systems described by differential equations, of which both the state and action spaces are continuous, and it is theoretically shown that at each iteration of ADP, a stable sub-optimal controller with improved performance can be obtained. Besides, the sequence of these sub-optimal controllers converges to the optimal one \cite{Book_Jiang}. Therefore, with these advantages over conventional RL, ADP has been applied in various fields, including autonomous driving\cite{Cui2021ACC}, transportation\cite{Huang2021}, and robotics \cite{Cui2021IROS}, where the safety of these systems plays a pivotal role.

Recently, based on the ADP technique, both the policy iteration (PI) and value iteration (VI) approaches are developed for various important classes of linear/nonlinear/periodic dynamical systems and for optimal stabilization, tracking and output regulation problems \cite{Jiang_tutorial,Jiang2012, Bian2016,Bian2021,Gao2016,Pang2021}. However, a systematic ADP approach to adaptive optimal control of continuous-time time-delay systems is lacking, due to the infinite-dimensional nature of these systems. In \cite{Huang_bookchapter}, although the model-free data-driven control for continuous time-delay systems is studied, discretization and/or linearization techniques are applied to transfer the infinite-dimensional system to a finite-dimensional delay-free system with augmented states, which leads to an approximate model. In \cite{Lin2019, Wei2010, Liu2016, Zhang2014, Zhang2011, Wang2014, Huang2021,Escobedo2021}, ADP for discrete-time systems with time delays is studied. Due to the finite dimensionality of discrete-time systems with time delays, these proposed ADP methods are not applicable to continuous time-delay systems. In \cite{Moghadam2021_bookchapter, Moghadam2021}, ADP technique is applied for both linear and nonlinear systems with time delay, but as mentioned in \cite[Remark 9.1]{Moghadam2021_bookchapter}, the integral term is not included in the design of ADP controller to avoid solving the ARE in infinite-dimensional space. As a consequence, the resulting controller can only stabilize the system without achieving optimality. Technically, there are several obstacles in the generalization of ADP to time-delay systems. Firstly, for an infinite-dimensional system, optimality properties are hard to analyze, because the corresponding ARE is complex partial differential equations (PDEs). Secondly, stability analysis and controller design for a time-delay system are much more challenging than the corresponding ones in the finite-dimensional system setting. Therefore, the model-free optimal control for a continuous time-delay system remains an open problem. 

In this paper, in the absence of the precise knowledge of system dynamics, a novel data-driven PI approach for continuous linear time-delay systems are proposed based on ADP. The contributions of this paper are as follows. Firstly, inspired by Kleinman's model-based PI algorithm for delay-free linear systems \cite{Kleinman1968}, a new model-based PI algorithm is proposed for a class of linear time-delay systems. Given an admissible initial controller, both the stability of the updated sub-optimal controller at each iteration and the convergence of the sequence of learned controllers to the (unknown) optimal controller are proved theoretically. It is worth pointing out that different from delay-free systems, due to the infinite dimensionality, both the value function and the control law for the linear time-delay systems are functional of the system's state, which in consequence increases the difficulty to design the PI algorithm. Secondly, based on the aforementioned model-based PI, this paper contributes a data-driven PI approach to adaptive optimal controller design using only the data measured along the trajectories of the system. 

The rest of this paper is organized as follows. Section \RNum{2} introduces the class of linear time-delay systems and formulates the adaptive optimal control problem to be addressed in the paper. Section \RNum{3} proposes a model-based PI approach to iteratively solve the LQ optimal control problem for linear time-delay systems. In Section \RNum{4}, based on the theoretical result of the previous section, a data-driven PI approach is proposed, and the convergence property of the algorithm is analyzed. Section \RNum{5} illustrates the proposed data-driven PI approach by means of two practical examples. Finally, some concluding remarks are drawn in Section \RNum{6}.

\textit{Notations:} In this paper, $\mathbb{R}$ denotes the set of real numbers, $\mathbb{R}_+$ denotes the set of nonnegative real numbers, and $\mathbb{N}_+$ denotes the set of positive integers. $|\cdot|$ denotes the Euclidean norm of a vector or Frobenius norm of a matrix. $\norm{\cdot}_\infty$ denotes the supremum norm of a function. $\mathcal{C}^0\left(X, Y \right)$ denotes the class of continuous functions from the linear space $X$ to the linear space $Y$, respectively. $\mathcal{AC}\left([-\tau,0], \mathbb{R}^n \right)$ denotes the class of absolutely continuous functions. $\frac{\mathrm{d}{f}}{\mathrm{d}{\theta}}(\cdot)$ denotes the function which is the derivative of the function $f$. $\oplus$ denotes the direct sum. $L_i([-\tau, 0], \mathbb{R}^n)$ denotes the space of measurable functions for which the $i$th power of the Euclidean norm is Lebesgue integrable, $\mathcal{M}_2 = \mathbb{R}^n \oplus {L}_2([-\tau, 0], \mathbb{R}^n)$, and $\mathcal{D} = \left\{ \begin{bmatrix} r \\ f(\cdot) \end{bmatrix} \in \mathcal{M}_2: f \in \mathcal{AC}, \frac{\mathrm{d}f}{\mathrm{d}\theta}(\cdot) \in L_2\text{, and } f(0) = r\right\}$. $\langle \cdot, \cdot \rangle$ denotes the inner product in $\mathcal{M}_2$, i.e. $\langle z_1, z_2 \rangle = r_1^\top r_2+ \int_{-\tau}^{0} f_1^\top(\theta)f_2(\theta)\mathrm{d}\theta$, where $z_i = [r_i,f_i(\cdot)]^\top$ for $i=1,2$. $\mathcal{L}(X)$ and $\mathcal{L}(X,Y)$ denote the class of continuous bounded linear operators from $X$ to $X$ and from $X$ to $Y$ respectively. $\otimes$ denotes the Kronecker product. $\text{vec}(A) = \left[a_1^\top,a_2^\top,...,a_n^\top \right]^\top$, where $A \in \mathbb{R}^{n \times n}$ and $a_i$ is the $i$th column of $A$. For a symmetric matrix $P \in \mathbb{R}^{n \times n}$, $\text{vecs}(P) = [p_{11},2p_{12},...,2p_{1n},p_{22},2p_{23},...,2p_{(n-1)n},p_{nn}]^\top$, $\text{vecu}(P) = [2p_{12},...,2p_{1n},2p_{23},...,2p_{(n-1)n}]^\top$, and $\text{diag}(P) = [p_{11},p_{22},...,p_{nn}]^\top$. For two arbitrary vectors $\nu, \mu \in \mathbb{R}^{n}$, $\text{vecd}(\nu,\mu) = [\nu_1\mu_1,\cdots, \nu_{n} \mu_n]^\top$, $\text{vecv}(\nu) = [\nu_1^2,\nu_1 \nu_2,...,\nu_1 \nu_n,\nu_2^2,...,\nu_{n-1} \nu_n,\nu_{n}^2]^\top$, $\text{vecp}(\nu,\mu) = [\nu_1 \mu_2,...,\nu_1 \mu_n, \nu_2 \mu_3,...,\nu_{n-1} \mu_n]^\top$.  $[a]_{i \sim j}$ denotes the sub-vector of the vector $a$ comprised of the entries between the $i$th and $j$th entries. $A^\dagger$ denotes the Moore-Penrose inverse of matrix $A$.


\section{Problem Formulation and Preliminaries}
\subsection{Problem Formulation}
Consider a linear time-delay system
\begin{align}
    \Dot{x}(t) = Ax(t) + A_dx(t-\tau) + Bu(t),
\label{eq: Time delay system}
\end{align}
where $\tau \in \mathbb{R}_+$ denotes the delay of the system and is assumed to be constant and known, $x(t) \in \mathbb{R}^n$, $u(t) \in \mathbb{R}^{m}$. $A$, $A_d \in \mathbb{R}^{n\times n}$ and $B \in \mathbb{R}^{n \times m}$ are unknown constant matrices. Let $x_t (\theta)=x(t+\theta), \; \forall \theta\in [-\tau, 0] $ denote a segment of the state trajectory in the interval $[t-\tau, t]$. Due to the infinite dimensionality of the system (\ref{eq: Time delay system}), the state of the system is $z(t) = [x^\top(t), x^\top_t(\cdot)]^\top \in \mathcal{M}_2$. Define the linear operators $\mathbf{A}\in \mathcal{L}(\mathcal{M}_2), \mathbf{B} \in \mathcal{L}(\mathbb{R}^m, \mathcal{M}_2)$ as $\mathbf{A}z(t) = \begin{bmatrix} Ax(t) + A_dx_t(-\tau) \\ \frac{\mathrm{d}x_t}{\mathrm{d}\theta}(\cdot)  \end{bmatrix}$ and  $\mathbf{B}u(t) = \begin{bmatrix} Bu(t) \\ 0  \end{bmatrix}$, where $\frac{\mathrm{d}x_t}{\mathrm{d}\theta}(\cdot)$ denotes the derivative of $x_t(\theta)$ with respect to $\theta$. Then, according to \cite[Theorem 2.4.6]{Curtain1995}, (\ref{eq: Time delay system}) can be rewritten as
\begin{align}
    \dot{z}(t) = \mathbf{A}z(t) + \mathbf{B} u(t), 
    \label{eq: linear time delay op}
\end{align}
with the domain of $\mathbf{A}$ given by $\mathcal{D}$.
Let $z_0 = \begin{bmatrix} x(0) \\ x_0(\cdot) \end{bmatrix} \in \mathcal{D}$ denote the initial state of the system (\ref{eq: linear time delay op}). The performance index of (\ref{eq: Time delay system}) is 
\begin{align}
\begin{split}
    J(x_0,u) &= \int_{0}^{\infty} x(t)^\top Q x(t) + u(t)^\top R u(t) \mathrm{d}t \\
    &=\int_{0}^{\infty} \langle z(t), \mathbf{Q} z(t) \rangle + u(t)^\top R u(t) \mathrm{d}t,
\end{split}
\label{eq: cost function}
\end{align}
where $R^\top = R > 0$, $Q^\top = Q \geq 0$, and $\mathbf{Q} = \begin{bmatrix}Q & \\ & \mathbf{0} \end{bmatrix} \in \mathcal{L}(\mathcal{M}_2)$ is symmetric \cite[Chapter 6]{Eidelman_book} and non-negative\cite[Definition 6.3.1]{Eidelman_book}. 
\begin{definition}\label{def:admissible}
For the system (\ref{eq: Time delay system}), a control policy $u_c(x_t): \mathcal{D} \rightarrow \mathbb{R}^m$ is called admissible with respect to (\ref{eq: cost function}), if the linear time-delay system (\ref{eq: Time delay system}) with $u = u_c(x_t)$ is globally asymptotically stable (GAS) at the origin \cite[Definition 1.1]{Book_Keqin}, and the performance index (\ref{eq: cost function}) is finite for all $z_0 \in \mathcal{D}$.
\end{definition}
\begin{assumption} \label{ass: Sta and Dec}
The system (\ref{eq: Time delay system}) with the output $y(t) = {Q}^{\frac{1}{2}}x(t)$ is exponentially stabilizable and detectable, which are defined in \cite[Definition 5.2.1]{Curtain1995} and can be checked according to \cite[Theorem 5.2.12]{Curtain1995}.
\end{assumption} 
\begin{remark}
    Assumption \ref{ass: Sta and Dec} is a standard prerequisite for LQ optimal control of system \eqref{eq: Time delay system} to ensure the existence of a stabilizing solution \cite{Curtain1995,Fridman2014}.
\end{remark}
Given the aforementioned assumption, the problems to be studied in this paper can be formulated as follows.

\begin{problem}
    Given an initial admissible controller $u_1(x_t) = -K_{0,1}x(t) - \int_{-\tau}^{0}K_{1,1}(\theta)x_t(\theta)\mathrm{d}\theta$, and without knowing the dynamics of the system (\ref{eq: Time delay system}), design a PI-based ADP algorithm to find approximations of the optimal controller which can minimize  (\ref{eq: cost function}) using only the input-state data measured along the trajectories of the system.   
\label{pbm: PI}
\end{problem}

\subsection{Optimality and Stability}
For a linear system without time delay, i.e. $A_d=0$ in (\ref{eq: Time delay system}), one can calculate the optimal controller by solving the ARE as discovered by Kalman \cite{kalman1960}. Correspondingly, for the linear time-delay system \eqref{eq: Time delay system}, the sufficient condition for a model-based solution to the optimal control problem is stated as follows.
\begin{lemma}[\cite{Ross1969optimal,Uchida1986}] \label{lm: Sufficient Optimality}

For system (\ref{eq: Time delay system}) with Assumption \ref{ass: Sta and Dec},
\begin{align}
\begin{split}
    &u^*(x_t) = -K_0^*x(t) - \int_{-\tau}^{0}K_1^*(\theta)x_t(\theta)\mathrm{d}\theta\\
    &= -R^{-1}B^\top P_0^{*}x(t) - \int_{-\tau}^{0}R^{-1}B^\top P^*_1(\theta)x_t(\theta)\mathrm{d}\theta 
\end{split}
\label{eq: optimal control}
\end{align}
is the optimal controller minimizing the cost (\ref{eq: cost function}), and the corresponding minimal performance index is
\begin{align}
\begin{split}
    V^*(x_0) &= x^\top(0)P^*_0x(0) + 2x^\top(0) \int_{-\tau}^{0} P^*_1(\theta)x_0(\theta)\mathrm{d}\theta \\
    &+ \int_{-\tau}^{0}\int_{-\tau}^{0} x_0^\top(\xi)P^*_2(\xi,\theta)x_0(\theta)\mathrm{d}\xi \mathrm{d} \theta,
\end{split} 
\label{eq: optimal cost}
\end{align}
where $P^*_0 = P_0^{*\top }> 0$, $P^*_1(\theta)$, and $P_2^{*\top}(\theta,\xi) = P^*_2(\xi,\theta)$ for $\theta,\xi \in [-\tau, 0]$ are the unique solution to the following PDEs 
\begin{align} 
\begin{split}
    & A^\top P^*_0 + P^*_0 A - P^*_0 B R^{-1} B^\top P^*_0 \\
    & \qquad \qquad \qquad \qquad \quad + P^*_1(0) + P_1^{*\top}(0) + Q = 0 ,\\
    & \frac{\mathrm{d}P^*_1(\theta)}{\mathrm{d}\theta} = (A^\top - P^*_0BR^{-1}B^\top)P^*_1(\theta) + P^*_2(0,\theta) ,\\
    & \partial_\xi P^*_2(\xi,\theta) + \partial_\theta P^*_2(\xi,\theta) = -P_1^{*\top}(\xi)BR^{-1}B^\top P^*_1(\theta) ,\\
    & P^*_1(-\tau) = P^*_0A_d ,\\
    & P^*_2(-\tau,\theta) = A_d^\top P^*_1(\theta). 
\end{split}
\label{eq: Ricatti Equation}
\end{align}
\end{lemma}


According to \cite[Theorem 6.2.7]{Curtain1995}, the time-delay system \eqref{eq: Time delay system} in closed-loop with $u^*$ is exponentially stable at the origin.


\section{Model-Based Policy Iteration}
According to Lemma \ref{lm: Sufficient Optimality}, if (\ref{eq: Ricatti Equation}) can be solved, the optimal controller is obtained. However, due to the non-linearity with respect to $P_0^*$, $P_1^*$ and $P_2^*$, it is difficult to solve (\ref{eq: Ricatti Equation}) directly. Therefore, the model-based PI algorithm is proposed to simplify the process of solving (\ref{eq: Ricatti Equation}). 

Given an admissible controller $u_1(x_t) = -K_{0,1}x(t) - \int_{-\tau}^{0} K_{1,1}(\theta) x_t(\theta)\mathrm{d}\theta$, the model-based PI algorithm for system (\ref{eq: Time delay system}) is proposed as follows.  
\begin{enumerate}
    \item \textit{Policy Evaluation}: For $i \in \mathbb{N}_+$, and $\xi,\theta \in [-\tau,0]$, calculate $P_{0,i} = P_{0,i}^{\top }> 0$, $P_{1,i}(\theta)$, and $P_{2,i}^{\top}(\theta,\xi) = P_{2,i}(\xi,\theta)$ by solving the following PDEs,
\begin{align}
\begin{split}
    & A_i^{\top} P_{0,i} + P_{0,i} A_i + Q_i + P_{1,i}(0) + P_{1,i}^\top(0)  = 0 ,\\
    & \frac{\mathrm{d}P_{1,i}(\theta)}{\mathrm{d}\theta} = A_i^{\top} P_{1,i}(\theta) - P_{0,i}B K_{1,i}(\theta) \\
    &\qquad \qquad \qquad \qquad  +  K^\top_{0,i}R K_{1,i}(\theta) + P_{2,i}(0,\theta) ,\\
    & \partial_\xi P_{2,i}(\xi,\theta) + \partial_\theta P_{2,i}(\xi,\theta) = K^\top_{1,i}(\xi)R K_{1,i}(\theta)\\
    &\qquad \qquad \qquad \qquad \qquad  \quad- 2K^\top_{1,i}(\xi)B^\top P_{1,i}(\theta) ,\\
    & P_{1,i}(-\tau) = P_{0,i}A_d ,\\
    & P_{2,i}(-\tau,\theta) = A_d^\top P_{1,i}(\theta) ,   
\label{eq: policy evaluation}
\end{split}    
\end{align}
where $A_i = (A - BK_{0,i})$ and $Q_i = Q + K^\top_{0,i} R K_{0,i}$. 
    \item \textit{Policy Improvement}: Update the policy $u_{i+1}$ by
    \begin{small}
    \begin{align}
    \begin{split}
        u_{i+1}(x_t) = -K_{0,i+1}x(t) - \int_{-\tau}^{0} K_{1,i+1}(\theta) x_t(\theta)\mathrm{d}\theta\\
        = -R^{-1}B^\top P_{0,i}x(t) - \int_{-\tau}^{0}R^{-1}B^\top P_{1,i}(\theta)x_t(\theta)\mathrm{d}\theta.
    \end{split}
    \label{eq: policy improvement}
    \end{align}
    \end{small}
\end{enumerate}

The policy evaluation calculates the value functional $V_i(x_0) = J(x_0, u_i)$, which is expressed as 
\begin{align}
\begin{split}
    V_i(x_0) &= x^\top(0)P_{0,i}x(0) + 2x^\top(0) \int_{-\tau}^{0} P_{1,i}(\theta)x_0(\theta)\mathrm{d}\theta \\
    &+ \int_{-\tau}^{0}\int_{-\tau}^{0} x_0^\top(\xi)P_{2,i}(\xi,\theta)x_0(\theta)\mathrm{d}\xi d \theta.
\end{split} 
\label{eq: Vi}
\end{align}
By policy improvement, the value functional is monotonically decreasing ($V_{i+1}(x_0) \leq V_{i}(x_0)$), and converges to the optimal value functional $V^*(x_0)$. Correspondingly,  $P_{0,i}$, $P_{1,i}(\theta)$ and $P_{2,i}(\xi,\theta)$ converge to the optimal solutions $P_{0}^*$, $P_{1}^*(\theta)$ and $P_{2}^*(\xi,\theta)$, respectively. The convergence of the model-based PI algorithm is rigorously demonstrated in Theorem \ref{lm: model-based PI}. Before stating Theorem \ref{lm: model-based PI}, we first introduce Lemma \ref{lm: stability} which is instrumental for the proof of Theorem \ref{lm: model-based PI}. With the help of Lemma \ref{lm: stability}, if the cost $J(x_0,u_L)$ for a linear controller $u_L$ is finite, the closed-loop system with $u_L$ is globally exponentially stable. Consequently, $u_L$ is admissible.

\begin{figure*}[b]
\noindent\makebox[\linewidth]{\rule{\textwidth}{0.4pt}}
\vspace{0pt}
\begin{small}
\begin{subequations}\label{eq: Vdot}
\begin{align}
    &\dot{V}_i(x_t) = (Ax(t) + A_d x(t-\tau) + Bu(t))^\top P_{0,i}x(t) + x^\top(t)P_{0,i}(Ax(t) + A_d x(t-\tau) + Bu(t)) \nonumber \\
    &\quad + 2(Ax(t) + A_d x(t-\tau) + Bu(t))^\top \int_{-\tau}^{0} P_{1,i}(\theta)x_t(\theta)\mathrm{d}\theta  + 2x^\top(t)\int_{-\tau}^{0}P_{1,i}(\theta)\frac{\mathrm{d}}{\mathrm{d}t}x(t+\theta)\mathrm{d}\theta  \nonumber\\
    &\quad+ \int_{-\tau}^{0}\int_{-\tau}^{0} \frac{\mathrm{d}}{\mathrm{d}t}x^\top(t+\xi)P_{2,i}(\xi,\theta)x(\theta)\mathrm{d}\xi \mathrm{d}\theta + \int_{-\tau}^{0}\int_{-\tau}^{0} x_t^\top(\xi)P_{2,i}(\xi,\theta)\frac{\mathrm{d}}{\mathrm{d}t}x(t+ \theta)\mathrm{d}\xi \mathrm{d}\theta, \\
    &= x^\top(t)(A^\top P_{0,i} + P_{0,i}A)x(t) + 2x^\top(t)P_{0,i}A_dx(t - \tau) + 2x^\top(t)P_{0,i}Bu(t) \nonumber\\
    &\quad+ 2(Ax(t) + A_d x(t-\tau) + Bu(t))^\top \int_{-\tau}^{0} P_{1,i}(\theta)x_t(\theta)\mathrm{d}\theta  + 2x^\top(t)P_{1,i}(\theta)x_t(\theta)|_{\theta=-\tau}^{0} - 2x^\top(t) \int_{-\tau}^{0}\frac{\mathrm{d}}{\mathrm{d}\theta}(P_{1,i}(\theta))x_t(\theta)\mathrm{d}\theta  \nonumber\\
    &\quad+ \int_{-\tau}^{0}x_t^\top(\xi)P_{2,i}(\xi,\theta)x_t(\theta) \mathrm{d}\theta |_{\xi=-\tau}^{0} - \int_{-\tau}^{0}\int_{-\tau}^{0}x_t^\top(\xi) \partial_\xi P_{2,i}(\xi,\theta)x(\theta)\mathrm{d}\xi \mathrm{d}\theta \nonumber\\
    &\quad+ \int_{-\tau}^{0} x_t^\top(\xi)P_{2,i}(\xi,\theta)x_t( \theta) \mathrm{d}\xi |_{\theta=-\tau}^{0} - \int_{-\tau}^{0}\int_{-\tau}^{0} x_t^\top(\xi) \partial_\theta P_{2,i}(\xi,\theta)x_t(\theta) \mathrm{d}\theta \mathrm{d}\xi \label{eq:partialInt}\\
    &= x^\top(t)(A^\top P_{0,i} + P_{0,i}A + P_{1,i}(0) + P_{1,i}^\top(0))x(t) + 2x^\top(t)(P_{0,i}A_d - P_{1,i}(-\tau))x(t - \tau) \nonumber\\
    &\quad+ 2x^\top(t)\int_{-\tau}^{0} \left(A^TP_{1,i}(\theta) - \frac{\mathrm{d}}{\mathrm{d}\theta}P_{1,i}(\theta) + P_{2,i}(0,\theta) \right)x_t(\theta)\mathrm{d}\theta + 2x^\top(t-\tau)\int_{-\tau}^{0} \left(A_d^\top P_{1,i}(\theta) - P_{2,i}(-\tau,\theta) \right)x_t(\theta)\mathrm{d}\theta \nonumber\\ 
    &\quad-\int_{-\tau}^{0}\int_{-\tau}^{0} x_t^\top(\xi)\left(\partial_\theta P_{2,i}(\xi,\theta) + \partial_\xi P_{2,i}(\xi,\theta) \right)x_t(\theta) \mathrm{d}\xi \mathrm{d}\theta + 2 u^\top(t) B^\top P_{0,i}x(t) + 2 u^\top(t)\int_{-\tau}^{0} B^\top P_{1,i}(\theta)x_t(\theta)\mathrm{d}\theta , \label{eq:CombSimilarTerms} \\
    &=  x^\top(t)(-Q_i + K_{0,i}^\top B^\top P_{0,i} + P_{0,i}B K_{0,i})x(t) + 2x^\top(t)\int_{-\tau}^{0} \left( K_{0,i}^\top B^\top P_{1,i}(\theta) + P_{0,i}BK_{1,i}(\theta) - K_{0,i}^\top R K_{1,i}(\theta) \right)x_t(\theta)\mathrm{d}\theta  \nonumber\\ 
    &\quad -\int_{-\tau}^{0}\int_{-\tau}^{0} x_t^\top(\xi)\left(K^\top_{1,i}(\xi)R K_{1,i}(\theta)
    - 2K^\top_{1,i}(\xi)B^\top P_{1,i}(\theta) \right)x_t(\theta) \mathrm{d}\xi \mathrm{d}\theta + 2 u^\top(t) \left(B^\top P_{0,i}x(t) + \int_{-\tau}^{0} B^\top P_{1,i}(\theta)x_t(\theta)\mathrm{d}\theta \right) \label{eq:simplyPI}\\
    & = -x^\top Q x - u_i^\top Ru_i + 2u_{i+1}^\top R u_{i} - 2u^\top R u_{i+1} . \label{eq:simplyuExp} 
\end{align}
\end{subequations}
\end{small}%
\end{figure*}

\begin{lemma} \label{lm: stability}
Consider the linear time-delay system (\ref{eq: Time delay system}) under Assumption \ref{ass: Sta and Dec}. If the linear controller $u_L(x_t) = -\mathbf{K}z(t)$ satisfies $J(x_0,u_L)<\infty$ for any $z_0 \in \mathcal{D}$, where $\mathbf{K} \in \mathcal{L}(\mathcal{M}_2, \mathbb{R}^m)$, then the closed-loop system with $u_L$ is globally exponentially stable at the origin.
\end{lemma}
\begin{proof}
Define $\mathbf{C} = [{Q}^{\frac{1}{2}}, \mathbf{0}] \in \mathcal{L}(\mathcal{M}_2,\mathbb{R}^n)$, and then $y(t) = {Q}^{\frac{1}{2}}x(t) = \mathbf{C}z(t)$. Due to the exponential detectability of $(\mathbf{A}, \mathbf{C})$, there exists $\mathbf{F} \in \mathcal{L}(\mathbb{R}^n, \mathcal{M}_2)$, such that $\mathbf{A}-\mathbf{F}\mathbf{C}$ is exponentially stable. Define $\mathbf{T}_{BK}(t)$ and $\mathbf{T}_{FC}(t)$ as the semigroups for $\mathbf{A}-\mathbf{B}\mathbf{K}$ and $\mathbf{A}-\mathbf{F}\mathbf{C}$ respectively. $J(x_0,u_L)<\infty$ implies that 
\begin{align}
\begin{split}
&\int_{0}^{\infty} |\mathbf{CT}_{BK}(t)z_0|^2\mathrm{d}t<\infty,\\
&\int_{0}^{\infty} |{R}^{\frac{1}{2}}\mathbf{KT}_{BK}(t)z_0|^2\mathrm{d}t<\infty.   
\end{split}
\label{eq: L2 of CT and BK}
\end{align}
According to \cite[Theorem 3.2.1]{Curtain1995}, we have 
\begin{align}
\begin{split}
    \mathbf{T}_{BK}&(t)z_0 = \mathbf{T}_{FC}(t)z_0 +\\
    & \int_{0}^{t}\mathbf{T}_{FC}(t-w)(\mathbf{FC-BK})\mathbf{T}_{BK}(w)z_0\mathrm{d}w.
\end{split}
\end{align}
Taking norm of the above equation yields
\begin{align}
\begin{split}
    &\|\mathbf{T}_{BK}(t)z_0\| \leq \|\mathbf{T}_{FC}(t)z_0\| + \int_{0}^{t}\|\mathbf{T}_{FC}(t-w) [\|\mathbf{F}\|\\
    &\quad |\mathbf{CT}_{BK}(w)z_0| + \|\mathbf{B}R^{-\frac{1}{2}}\| |{R}^{\frac{1}{2}}\mathbf{KT}_{BK}(w)z_0 |]\mathrm{d}w.
\end{split}
\end{align}
By (\ref{eq: L2 of CT and BK}), $|\mathbf{CT}_{BK}(\cdot)z_0|,|{R}^{\frac{1}{2}}\mathbf{KT}_{BK}(\cdot)z_0| \in L_2([0,\infty),\mathbb{R})$. Furthermore, since $\mathbf{T}_{FC}(t)$ is exponentially stable, $\|\mathbf{T}_{FC}(\cdot)\| \in L_1([0,\infty),\mathbb{R}) \cap L_2([0,\infty),\mathbb{R})$. Hence, according to \cite[Lemma A6.6]{Curtain1995}, $\|\mathbf{T}_{BK}(\cdot)z_0\| \in L_2([0,\infty),\mathbb{R})$, which implies that $\mathbf{T}_{BK}(t)$ is globally exponentially stable \cite[Lemma 5.1.2]{Curtain1995}.
\end{proof}

\begin{theorem}
Given the admissible control $u_1(x_t)$, for $P_{0,i}$,  $P_{1,i}(\theta)$, $P_{2,i}(\xi,\theta)$, and $u_{i+1}(x_t)$ obtained by solving (\ref{eq: policy evaluation}) and (\ref{eq: policy improvement}), and for all $i \in \mathbb{N}_+$, the following properties hold.
\begin{enumerate}
    \item $V^*(x_0) \leq V_{i+1}(x_0) \leq V_{i}(x_0)$;
    \item $u_{i+1}(x_t)$ is admissible;
    \item $V_i(x_0)$ and $u_i(x_t)$ converge to $V^*(x_0)$ and $u^*(x_t)$ respectively.
\end{enumerate}
\label{thm: model-based PI}
\end{theorem}
\begin{proof}
Along the trajectories of \eqref{eq: Time delay system}, $\dot{V}_i(x_t)$ is derived in \eqref{eq: Vdot}. In detail, \eqref{eq:partialInt} is derived by partial integration and the fact that $\frac{\mathrm{d}}{\mathrm{d}t}x(t+\theta) = \frac{\mathrm{d}}{\mathrm{d}\theta}x(t+\theta) = \frac{\mathrm{d}}{\mathrm{d}\xi}x(t+\xi)$, \eqref{eq:CombSimilarTerms} is derived by combining the similar items and the fact that $P_{2,1}^\top(\xi,\theta) = P_{2,1}(\theta,\xi)$, \eqref{eq:simplyPI} is derived by plugging \eqref{eq: policy evaluation} into \eqref{eq:CombSimilarTerms}, and \eqref{eq:simplyuExp} is derived by substituting the expression of $u_i$ and \eqref{eq: policy improvement} into \eqref{eq:simplyPI}. The properties 1) and 2) are proved by induction.

When $i = 1$ and system (\ref{eq: Time delay system}) is driven by the admissible control $u_1(x_t)$, according to (\ref{eq: Vdot}), we have
\begin{align}\label{eq:V1dot}
    \dot{V}_1(x_t) =  -x^\top Q x - u_1^\top Ru_1.
\end{align}
Since $u_1$ is admissible, $x(\infty)=0$, and integrating \eqref{eq:V1dot} from $0$ to $\infty$ yields
\begin{align}
    V_1(x_0) &= \int_{0}^{\infty} x^\top(t)Qx(t) + u^\top_1(t)Ru_1(t) \mathrm{d}t \nonumber \\
    &= J(x_0,u_1)<\infty.
\end{align}
Along the trajectories of (\ref{eq: Time delay system}) driven by $u_2(x_t)$, by (\ref{eq: Vdot}), we have 
\begin{align}
\begin{split}
    \Dot{V}_1&(x_t) = -x^\top Q x - u_1^\top R u_1 - 2u^\top_2R(u_2- u_1) \\
    &= -x^\top Q x - u_2^\top R u_2 -(u_2 - u_1)^\top R (u_2 - u_1).
\end{split}
\label{eq: V1Dotu2}
\end{align}
Integrating both sides of (\ref{eq: V1Dotu2}) from $0$ to $\infty$ yields
\begin{small}
\begin{align}\label{eq:V2finite}
    J(x_0, u_2) &= V_1(x_0) - V_1(x_{\infty}) - \int_{0}^{\infty}(u_2 - u_1)^\top R (u_2 - u_1)\mathrm{d}t \nonumber\\
    &\leq V_1(x_0)<\infty
\end{align}
\end{small}%
Since $J(x_0, u_2)$ is finite, by Lemma \ref{lm: stability}, $u_2$ is a globally exponentially stabilizing controller. Consequently, by Definition \ref{def:admissible}, $u_2$ is admissible. By \eqref{eq: Vdot}, when system \eqref{eq: Time delay system} is driven by $u_2$,
\begin{align}\label{eq:V2dot}
    \Dot{V}_2(x_t) = -x^\top Q x - u_2^\top R u_2.
\end{align} 
Since $u_2$ is admissible, integrating \eqref{eq:V2dot} from $0$ to $\infty$, we have $J(x_0,u_2) = V_2(x_0)$. Therefore, from \eqref{eq:V2finite}, we have $V_2(x_0) \leq V_1(x_0)$.

When $i>1$, assume 1) and 2) hold. When system \eqref{eq: Time delay system} is driven by $u_i$, the expression of $\dot{V}_i(x_t)$ is
\begin{align} \label{eq:Vidot_ui}
    \dot{V}_i(x_t) = -x^\top Q x - u_i^\top R u_i.
\end{align}
Noting the fact that $u_i$ is admissible and integrating \eqref{eq:Vidot_ui} from $0$ to $\infty$, we have 
\begin{align}
    V_i(x_0) &= \int_{0}^{\infty} x^\top(t)Qx(t) + u^\top_i(t)Ru_i(t) \mathrm{d}t \nonumber\\
    &= J(x_0,u_i) <\infty.
\end{align}
By \eqref{eq: Vdot}, along the state trajectories of (\ref{eq: Time delay system}) driven by $u_{i+1}$, 
\begin{align}
\begin{split}
    \Dot{V}_i(x_t) =& -x^\top Q x - u_{i+1}^\top R u_{i+1} \\& \qquad \qquad \qquad -(u_{i+1} - u_i)^\top R (u_{i+1} - u_i) .   
\end{split}
\label{eq: dot V_i along u_i+1}
\end{align}
Integrating both sides of (\ref{eq: dot V_i along u_i+1}) from $0$ to $\infty$ yields 
\begin{align}\label{Ji+1finite}
    &J(x_0, u_{i+1})=V_i(x_0) - V_i(x_\infty) \nonumber\\
    &\quad -\int_{0}^{\infty}(u_{i+1} - u_i)^\top R (u_{i+1} - u_i)\mathrm{d}t \leq V_i(x_0)<\infty.
\end{align}
Since the cost of $u_{i+1}$ is finite, by Lemma \ref{lm: stability}, $u_{i+1}$ is a globally exponentially stabilizing controller. Consequently, by Definition \ref{def:admissible}, $u_{i+1}$ is admissible. Along the state trajectories of system \eqref{eq: Time delay system} driven by $u_{i+1}$, by \eqref{eq: Vdot}, 
\begin{align}\label{eq:Vi+1dot}
    \Dot{V}_{i+1}(x_t) = -x^\top Q x - u^\top_{i+1} R u_{i+1}.
\end{align}
Since $u_{i+1}$ is admissible, integrating \eqref{eq:Vi+1dot} from $0$ to $\infty$ yields $V_{i+1}(x_0) = J(x_0,u_{i+1})$. 
Hence, $V_{i+1}(x_0) \leq V_i(x_0)$ is obtained by \eqref{Ji+1finite}. Furthermore, since $V^*(x_0) = J(x_0,u^*)$ is the minimal value of the performance index by Lemma \ref{lm: Sufficient Optimality}, for any $i \in \mathbb{N}_+$, $V^*(x_0) \leq V_{i}(x_0)$. Therefore, the proof of 1) and 2) is completed by induction.

With the model-based PI algorithm (\ref{eq: policy evaluation}) and (\ref{eq: policy improvement}), we have
\begin{align}
    V^*(x_0) \leq \cdots \leq V_i(x_0) \leq \cdots \leq V_1(x_0).
\label{eq: ineq of V}
\end{align}
Define $\mathbf{P}_i \in \mathcal{L}(\mathcal{M}_2)$, such that for any $z_0 = \begin{bmatrix}x_0(0) \\ x_0(\cdot) \end{bmatrix} \in \mathcal{D}$, $\mathbf{P}_i z_0$ can be expressed as
\begin{align}
    \mathbf{P}_i z_0 &= \begin{bmatrix}P_{0,i}x(0) + \int_{-\tau}^{0} P_{1,i}(\theta)x_0(\theta)\mathrm{d}\theta \\ \int_{-\tau}^{0} P_{2,i}(\cdot,\theta)x_0(\theta)\mathrm{d}\theta + P^\top_{1,i}(\cdot)x(0) \end{bmatrix}.
\label{eq: operator Pi}
\end{align}
It is easy to check that $ \mathbf{P}_i$ is symmetric, positive semi-definite, and $V_i(x_0) = \langle z_0, \mathbf{P}_i z_0 \rangle$. Furthermore, according to (\ref{eq: ineq of V}), for any $i \in \mathbb{N}_+$, $\mathbf{P}_{i+1} \leq \mathbf{P}_i\leq \mathbf{P}_{i-1}$. According to \cite[Theorem 6.3.2]{Eidelman_book}, there exists $\mathbf{P}_p = \mathbf{P}_p^\top > 0$, such that for all $z_0 \in \mathcal{M}_2$, we have
\begin{align}
    \lim_{i \to \infty}\mathbf{P}_i z_0 = \mathbf{P}_pz_0.
\end{align}
Therefore, $P_{0,i}$, $P_{1,i}(\theta)$ and $P_{2,i}(\xi,\theta)$ finally pointwisely converge to $P_{0,p}$, $P_{1,p}(\theta)$, and $P_{2,p}(\xi,\theta)$ respectively. When $\mathbf{P}_i$ converges, by the policy evaluation step \eqref{eq: policy evaluation}, $P_{0,p}$, $P_{1,p}(\theta)$ and $P_{2,p}(\xi,\theta)$ satisfy
\begin{align}
\begin{split}
    & (A - BK_{0,p})^{\top} P_{0,p} + P_{0,p} (A - BK_{0,p}) + Q \\
    &\quad+ K^\top_{0,p} R K_{0,p} + P_{1,p}(0) + P_{1,p}^\top(0)  = 0 ,\\
    & \frac{\mathrm{d}P_{1,p}(\theta)}{\mathrm{d}\theta} = (A - BK_{0,p})^{\top} P_{1,p}(\theta) - P_{0,p}B K_{1,p}(\theta) \\
    &\quad +  K^\top_{0,p}R K_{1,p}(\theta) + P_{2,p}(0,\theta) ,\\
    & \partial_\xi P_{2,p}(\xi,\theta) + \partial_\theta P_{2,p}(\xi,\theta) = K^\top_{1,p}(\xi)R K_{1,p}(\theta) \\
    &\quad - 2K^\top_{1,p}(\xi)B^\top P_{1,p}(\theta) ,\\
    & P_{1,p}(-\tau) = P_{0,p}A_d ,\\
    & P_{2,p}(-\tau,\theta) = A_d^\top P_{1,p}(\theta). 
\end{split}    
\label{eq: PolicyEvaluation_pth}
\end{align}
Since $P_{0,i}$, $P_{1,i}(\theta)$ and $P_{2,i}(\xi,\theta)$ converges to $P_{0,p}$, $P_{1,p}(\theta)$ and $P_{2,p}(\xi,\theta)$, respectively, $K_{0,i}$ and $K_{1,i}$ converge to $K_{0,p}$ and $K_{1,p}$. And by the policy improvement step \eqref{eq: policy improvement}, $K_{0,p}$ and $K_{1,p}$ satisfy
\begin{align}
     K_{0,p} = R^{-1}B^\top P_{0,p}, \qquad K_{1,p}(\theta) = R^{-1}B^\top P_{1,p}(\theta).
\label{eq:policyImprove_pth}
\end{align}
Substituting \eqref{eq:policyImprove_pth} into \eqref{eq: PolicyEvaluation_pth}, it is seen that $P_{0,p}$, $P_{1,p}(\theta)$ and $P_{2,p}(\xi,\theta)$ solve the PDEs \eqref{eq: Ricatti Equation}. Due to the uniqueness of the solution to (\ref{eq: Ricatti Equation}), $P_{0,i}$, $P_{1,i}(\theta)$ and $P_{2,i}(\xi,\theta)$ pointwisely converge to $P_{0}^*$, $P_{1}^*(\theta)$ and $P_{2}^*(\xi,\theta)$. Since both $P_{1,i}(\theta)$ and $P_{2,i}(\xi,\theta)$ are continuously differentiable, $ \{P_{1,i}(\theta): i \in \mathbb{N}_+\}$ and $\{P_{2,i}(\xi,\theta): i \in \mathbb{N}_+\}$ are equicontinuous, which leads to the uniform convergence by \cite[Chapter 4, Theorem 16]{Book_Pugh}. Hence, 3) can be proved.
\end{proof}

Notice that although (\ref{eq: policy evaluation}) is linear with respect to $P_{0,i}$, $P_{1,i}$, and $P_{2,i}$, due to the existence of PDEs, solving the analytical solution to (\ref{eq: policy evaluation}) is still non-trivial. Besides, the accurate knowledge of system matrices $A$, $A_d$, and $B$ is required to implement the model-based PI, and in practice due to the complex structure of the system, it is often hard to derive such an accurate model. Therefore, in the next section, a data-driven PI algorithm is proposed.

\begin{remark}
When $A_d = 0$, (\ref{eq: Time delay system}) is degraded to the normal linear time-invariant systems. According to (\ref{eq: policy evaluation}) and (\ref{eq: policy improvement}), we can see that $P_{1,i}(\theta)=0$, $P_{2,i}(\xi, \theta)=0$, and  $K_{1,i}(\theta)=0$. As a consequence, (\ref{eq: policy evaluation}) and (\ref{eq: policy improvement}) are same as the model-based PI method in \cite{Kleinman1968}. Therefore, the proposed model-based PI algorithm is a generalization of the celebrated Kleinman algorithm to linear time-delay systems.  
\end{remark}
\begin{remark}
In \cite{Burns2008}, the model-based PI is developed for linear infinite-dimensional systems in Hilbert space. Although the linear time-delay system is one of the infinite-dimensional systems, the concrete expression of PI for linear time-delay systems is not given in \cite{Burns2008}, and as a consequence, the PI developed in \cite{Burns2008} cannot be directly applied to solve the PDEs (\ref{eq: Ricatti Equation}). In this paper, the concrete expression of PI is constructed in (\ref{eq: policy evaluation}) and (\ref{eq: policy improvement}), which is one of the major contributions in this paper. Besides, it can be checked that at each iteration, $\mathbf{P}_i$ defined in (\ref{eq: operator Pi}) satisfies the PI update equations in \cite{Burns2008}, which is another way to prove the validity of the proposed PI theoretically.
\end{remark}
\begin{remark}
As shown in \cite{Burns2008}, the convergence rate of PI algorithm in the Hilbert space is quadratic, and therefore, the proposed model-based PI for time-delay systems has the same quadratic convergence rate.
\end{remark}
\section{Data-driven Policy Iteration}

The purpose of this section is to propose a corresponding data-driven PI method that does not require the accurate knowledge of system (\ref{eq: Time delay system}) to solve Problem 1. The input-state trajectory data of system \eqref{eq: Time delay system} is required for the data-driven PI, that is the continuous-time trajectories of $x(t)$ and $u(t)$ sampled from system \eqref{eq: Time delay system} within the interval $[t_1, t_{L+1}]$ is applied to train the control policy. In this section, $x(t)$ denotes the sampled state of system \eqref{eq: Time delay system} driven by the exploratory input $u(t)$.

Define $v_i(t) = u(t) - u_i(x_t)$. By \eqref{eq: Vdot}, along the trajectories of system (\ref{eq: Time delay system}) driven by $u$, 
\begin{align}\label{eq:Vdotvi}
    \dot{V}_i(x_t)= -x^\top Q x - u_i^\top R u_i -2u_{i+1}^\top R v_i.
\end{align}
Let $[t_k, t_{k+1} ]$ denote the $k$th segment of the interval $[t_1, t_{L+1}]$. Integrating both sides of \eqref{eq:Vdotvi} from $t_{k}$ to $t_{k+1}$ yields 
\begin{align}
\begin{split}
   &V_i(x_{t_{k+1}}) - V_i(x_{t_k}) = \\
   & \qquad \qquad \quad \int_{t_k}^{t_{k+1}} -x^\top Q x - u_i^\top R u_i - 2u_{i+1}^\top Rv_i \mathrm{d}t.
\end{split}
\label{eq: Diff Vi}
\end{align}
\begin{figure*}[b]
\noindent\makebox[\linewidth]{\rule{\textwidth}{0.4pt}}
\begin{align}
\begin{split}
    &\left[ x^\top(t)P_{0,i}x(t) + 2x^\top(t) \int_{-\tau}^{0} P_{1,i}(\theta)x_t(\theta)\mathrm{d}\theta + \int_{-\tau}^{0}\int_{-\tau}^{0} x_t^\top(\xi)P_{2,i}(\xi,\theta)x_t(\theta)\mathrm{d}\xi d \theta \right]_{t = t_k}^{t_{k+1}} \\
    &- 2\int_{t_k}^{t_{k+1}} \left(x^\top(t)K^\top_{0,i+1} + \int_{-\tau}^{0} x_t^\top(\theta)K^\top_{1,i+1}(\theta)\mathrm{d}\theta  \right)  Rv_i(t) \mathrm{d}t=  -\int_{t_k}^{t_{k+1}} x^\top Q x + u_i^\top R u_i \mathrm{d}t.
\end{split}    
\label{eq: Diff V with P}
\end{align}
\vspace{0pt}
\end{figure*}%
Plugging the expressions of $u_{i+1}$ in \eqref{eq: policy improvement} and $V_i$ in \eqref{eq: Vi} into (\ref{eq: Diff Vi}), one can obtain (\ref{eq: Diff V with P}), which is instrumental for the development of data-driven PI. 

As seen in \eqref{eq: policy evaluation} and \eqref{eq: policy improvement}, $K_{1,i}(\theta)$ and $P_{1,i}(\theta)$ are continuous functions defined on the interval $[-\tau,0]$; $P_{2,i}(\xi,\theta)$ is a continuous function defined over the set $[-\tau,0]^2$. Next, we will use the linear combinations of the basis functions to approximate these continuous functions, such that only the weighting matrices of the basis functions should be determined for the function approximation. Let $\Phi(\theta)$, $\Lambda(\xi, \theta)$, and $\Psi(\xi, \theta)$ denote the $N$-dimensional vectors of linearly independent basis functions. To simplify the notation, we choose the same number of basis functions for $\Phi$, $\Lambda$ and $\Psi$. According to the approximation theory \cite{powell_1981}, the following equations hold
\begin{align}
\begin{split}
    &\text{vecs}(P_{0,i}) = W_{0,i} ,\\
    &\text{vec}(P_{1,i}(\theta)) = W^N_{1,i} \Phi(\theta) + e^N_{\Phi,i}(\theta) ,\\
    &\text{diag}(P_{2,i}(\xi, \theta)) = W^N_{2,i} \Psi(\xi,\theta) + e^N_{\Psi,i}(\xi, \theta),\\
    &\text{vecu}(P_{2,i}(\xi, \theta)) = W^N_{3,i} \Lambda(\xi,\theta) + e^N_{\Lambda,i}(\xi, \theta) ,\\
    &\text{vec}(K_{0,i}) = U_{0,i} ,\\
    &\text{vec}(K_{1,i}(\theta)) = U^N_{1,i} \Phi(\theta) + e^N_{K,i}(\theta),
\end{split}    
\label{eq: approximate formula}
\end{align}
where $W_{0,i} \in \mathbb{R}^{n_1}$, $n_1 = \frac{n(n+1)}{2}$, $W^N_{1,i} \in \mathbb{R}^{n^2 \times N}$,  $W^N_{2,i} \in \mathbb{R}^{n \times N}$, $W^N_{3,i} \in \mathbb{R}^{n_2 \times N}$, $n_2 = \frac{n(n-1)}{2}$,  $U_{0,i} \in \mathbb{R}^{nm}$, and $U^N_{1,i} \in \mathbb{R}^{nm \times N}$ are weighting matrices of the basis functions. $e^N_{\Phi,i}(\theta) \in \mathcal{C}^0([-\tau,0], \mathbb{R}^{n^2})$, $e^N_{\Psi,i}(\xi, \theta) \in \mathcal{C}^0([-\tau,0]^2,\mathbb{R}^{n})$, $e^N_{\Lambda,i}(\xi, \theta) \in \mathcal{C}^0([-\tau,0]^2,\mathbb{R}^{n_2})$, and $e^N_{K,i}(\theta) \in \mathcal{C}^0([-\tau,0],\mathbb{R}^{mn})$ are approximation truncation errors. Therefore, according to the uniform approximation theory, as $N \rightarrow \infty$, the truncation errors converge uniformly to zero, i.e. for any $\eta>0$, there exists $N^*\in \mathbb{N}_+$, such that if $N>N^*$, the following inequalities hold 
\begin{align}
\begin{split}
    &\lVert e^N_{\Phi,i}(\theta)\lVert_\infty \leq \eta, \quad \quad   \lVert e^N_{K,i}(\theta)\lVert_\infty \leq \eta,\\
    & \lVert e^N_{\Psi,i}(\xi, \theta)\lVert_\infty \leq \eta, \quad  \lVert e^N_{\Lambda,i}(\xi, \theta)\lVert_\infty \leq \eta.
\end{split}
\label{eq: trun err}
\end{align}

Therefore, the key idea of data-driven PI is that $W_{j,i} (j=0,\cdots,3)$ and $U_{j,i} (j=0,1)$ are directly approximated by the data collected from system (\ref{eq: Time delay system}). Define $\Upsilon_i^N$ as the composite vector of the weighting matrices, i.e. 
\begin{align}
\begin{split}
        & \Upsilon^N_i = \left[W_{0,i}^\top , \text{vec}^\top(W^N_{1,i}), \text{vec}^\top(W^N_{2,i}), \text{vec}^\top(W^N_{3,i})\right.\\
        & \qquad \quad \left.U^\top_{0,i+1} ,\text{vec}^\top(U^N_{1,i+1})\right]^\top.
\end{split}
\label{eq: Upsilon}
\end{align}
Let $\hat{\Upsilon}^N_i$ be the approximation of $\Upsilon^N_i$, and then, the approximations of ${P}_{j,i} (j=0,1,2)$ can be reconstructed by
\begin{align}
\begin{split}
    &\hat{P}_{0,i} = \text{vec}^{-1}([\hat{\Upsilon}^N_i]_{1 \sim n_1}), \\
    &\hat{W}^N_{1,i} = \text{vec}^{-1}([\hat{\Upsilon}^N_i]_{n_1+1 \sim n_1+n^2N}),\\
    &\hat{W}^N_{2,i} = \text{vec}^{-1}([\hat{\Upsilon}^N_i]_{n_1+n^2N+1 \sim n_1+n^2N+nN}),\\
    &\hat{W}^N_{3,i} = \text{vec}^{-1}([\hat{\Upsilon}^N_i]_{n_1+n^2N+nN+1 \sim n_3-1}),\\
    &\hat{P}_{1,i}(\theta) = \text{vec}^{-1}(\hat{W}^N_{1,i} \Phi(\theta)),\\
    &\text{diag}(\hat{P}_{2,i}(\xi,\theta)) =  \hat{W}^N_{2,i} \Psi(\xi, \theta),\\
    &\text{vecu}(\hat{P}_{2,i}(\xi,\theta)) =  \hat{W}^N_{3,i} \Lambda(\xi, \theta),\\    
\end{split}
\label{eq: reconstruct P}
\end{align}
where $n_3 = n_1+n^2N+nN+n_2N+1$. Furthermore, $\hat{K}_{0,i+1}$ and $\hat{K}_{1,i+1}(\theta)$, the approximations of  ${K}_{0,i+1}$ and ${K}_{1,i+1}(\theta)$ respectively, can be reconstructed by
\begin{align}
\begin{split}
    &\hat{K}_{0,i+1} = \text{vec}^{-1}([\hat{\Upsilon}^N_i]_{n_3 \sim n_4}),\\
    &\hat{U}_{1,i+1} = \text{vec}^{-1}([\hat{\Upsilon}^N_i]_{n_4+1 \sim n_5}),\\
    &\hat{K}_{1,i+1}(\theta) = \text{vec}^{-1}(\hat{U}_{1,i}\Phi(\theta)),
\end{split}
\label{eq: reconstruct K}
\end{align}
where $n_4 = n_3+nm$, and $n_5 = n_4+nmN$. As a consequence, $\hat{u}_i(x_t)$, the approximation of $u_i(x_t)$, can be expressed as $\hat{u}_i(x_t) = -\hat{K}_{0,i}x(t) - \int_{-\tau}^{0} \hat{K}_{1,i}(\theta) x_t(\theta) \mathrm{d}\theta$.

Based on the approximations in \eqref{eq: approximate formula}, (\ref{eq: Diff V with P}) is transferred to a linear equation with respect to $\hat{\Upsilon}_{i}^{N}$. Then, the unknown vector ${\Upsilon}_{i}^{N}$ is solved by linear regression, and consequently, $P_{j,i}(j=0,1,2)$ and $K_{j,i+1}(j = 1,2)$ can be approximated by \eqref{eq: reconstruct P} and \eqref{eq: reconstruct K}. In detail, let $\hat{v}_i = u-\hat{u}_i$, $\tilde{u}_i = \hat{u}_i-u_i$. Define the data-constructed matrices $\Gamma_{\Phi xx}(t)$, $\Gamma_{\Psi xx}(t)$, $\Gamma_{\Lambda xx}(t)$, $G_{x\hat{v}_i,k}$, and $G_{\Phi x\hat{v}_i,k}$ as 
\begin{small}
\begin{align}\label{eq: definition of Gm}
    &\Gamma_{\Phi xx}(t) = \int_{-\tau}^{0} \Phi^\top(\theta) \otimes x^\top_t(\theta) \otimes x^\top(t)   \mathrm{d}\theta, \nonumber\\
    &\Gamma_{\Psi xx}(t) = \int_{-\tau}^{0}\int_{-\tau}^{0} \Psi^\top(\xi,\theta) \otimes \text{vecd}^\top(x_t(\xi), x_t(\theta)) 
    \mathrm{d}\xi \mathrm{d}\theta , \nonumber\\
    &\Gamma_{\Lambda xx}(t) = \int_{-\tau}^{0}\int_{-\tau}^{0} \Lambda^\top(\xi,\theta) \otimes \text{vecp}^\top(x_t(\xi), x_t(\theta)) 
    \mathrm{d}\xi \mathrm{d}\theta ,\\
    &G_{x\hat{v}_i,k} = \int_{t_k}^{t_{k+1}} (x^\top(t) \otimes \hat{v}_i^\top(t))(I_n \otimes R)\mathrm{d}t, \nonumber\\
    &G_{\Phi x\hat{v}_i,k} = \int_{t_k}^{t_{k+1}} \int_{-\tau}^{0} \Phi^\top(\theta) \otimes((x_t^\top(\theta) \otimes \hat{v}_i^\top(t)) (I_n \otimes R))  \mathrm{d}\theta \mathrm{d}t. \nonumber
\end{align}
\end{small}
\begin{figure*}[b]
\noindent\makebox[\linewidth]{\rule{\textwidth}{0.4pt}}
\vspace{0pt}
\begin{align}
\begin{split}
    &x^\top(t){P}_{0,i}x(t) = \text{vecv}^\top(x(t))W_{0,i},\\
    &x^\top(t) \int_{-\tau}^{0} P_{1,i}(\theta)x_t(\theta)d\theta =  \Gamma_{\Phi xx}(t)\text{vec}(W_{1,i}) + \int_{-\tau}^{0} x^\top_t(\theta) \otimes x^\top(t) e_{\Phi,i}^N(\theta)   d\theta = \Gamma_{\Phi xx}(t)\text{vec}(W_{1,i}) + \epsilon_{1,i}(t),\\
    &\int_{-\tau}^{0}\int_{-\tau}^{0} x^\top_t(\xi)P_{2,i}(\xi,\theta)x_t(\theta)d\xi d \theta =   \Gamma_{\Psi xx}(t) \text{vec}(W_{2,i}) + \Gamma_{\Lambda xx}(t) \text{vec}(W_{3,i}) \\
    &\qquad + \int_{-\tau}^{0}\int_{-\tau}^{0} \text{vecd}^\top(x_t(\xi), x_t(\theta)) e_{\Psi,i}^N(\xi, \theta) d\xi d\theta +  \int_{-\tau}^{0}\int_{-\tau}^{0}\text{vecp}^\top(x_t(\xi), x_t(\theta)) e_{\Lambda,i}^N(\xi,\theta)
    d\xi d\theta \\
    &\qquad = \Gamma_{\Psi xx}(t) \text{vec}(W_{2,i}) + \Gamma_{\Lambda xx}(t) \text{vec}(W_{3,i}) + \epsilon_{2,i}(t) + \epsilon_{3,i}(t),\\
    &\int_{t_k}^{t_{k+1}} x^\top(t)K^\top_{0,i+1} Rv_i(t) dt = G_{x\hat{v}_i,k} U_{0,i+1} + \int_{t_k}^{t_{k+1}} x^\top(t)K^\top_{0,i+1} R\tilde{u}_i(t) dt  = G_{x\hat{v}_i,k} U_{0,i+1} + \rho^0_{i,k} ,\\ 
    &\int_{t_k}^{t_{k+1}} \int_{-\tau}^{0} x_t^\top(\theta)K^\top_{1,i+1}(\theta)Rv_i(t) d\theta dt = G_{\Phi x\hat{v}_i,k} \text{vec}{(U_{1,i+1})} + \int_{t_k}^{t_{k+1}} \int_{-\tau}^{0}  (x_t^\top(\theta) \otimes v_i^\top(t))(I_n \otimes R) e_{K,i}^N(\theta)  d\theta dt \\
    & \qquad +\int_{t_k}^{t_{k+1}} \int_{-\tau}^{0} x_t^\top(\theta)K^\top_{1,i+1}(\theta)R\tilde{u}_i(t) d\theta dt = G_{\Phi x\hat{v}_i,k} \text{vec}{(U_{1,i+1})}  + \psi_{i,k} + \rho^1_{i,k}.
\end{split}    
\label{eq: kron expres}
\end{align}
\end{figure*}

With the help of (\ref{eq: approximate formula}) and (\ref{eq: definition of Gm}), each term in \eqref{eq: Diff V with P} is expressed linearly with respect to the weighting matrices in (\ref{eq: kron expres}). Then, the data collected from $L$ intervals along the trajectories of system (\ref{eq: Time delay system}) driven by $u$ will be applied to generate the adaptive optimal controller. Let $t_1<t_2<\cdots<t_L<t_{L+1}$ denote the boundaries of each interval.  With the collected data, the following variables are defined
\begin{align}
\begin{split}
    &M_{i,k} = \left[\text{vecv}^\top(x(t))|_{t_k}^{t_{k+1}} , 2\Gamma_{\Phi xx}|_{t_k}^{t_{k+1}}, \Gamma_{\Psi xx}(t)|_{t_k}^{t_{k+1}},\right.\\ 
    & \qquad \quad \left. \Gamma_{\Lambda xx}(t)|_{t_k}^{t_{k+1}}, -2G_{x\hat{v}_i,k} , -2G_{\Phi x\hat{v}_i,k}\right],\\
    &Y_{i,k} = -\int_{t_k}^{t_{k+1}} x^\top Q x + \hat{u}_i^\top R \hat{u}_i \mathrm{d}t, \\
    &E_{i,k} = \left[2\epsilon_{1,i}(t) + \epsilon_{2,i}(t) + \epsilon_{3,i}(t)\right]_{t = t_k}^{t_{k+1}} - 2\psi_{i,k} - 2\rho^0_{i,k} \\
    & \qquad \quad  -2\rho^1_{i,k}-\rho^2_{i,k}, \\
    & M_i = \left[M_{i,1}^\top,\cdots,M_{i,k}^\top, \cdots,M_{i,L}^\top \right]^\top, \\
    &Y_i = \left[Y_{i,1},\cdots,Y_{i,k},\cdots,Y_{i,L} \right] ^\top,\\
    &E_i = \left[E_{i,1},\cdots,E_{i,k},\cdots,E_{i,L} \right] ^\top,
\end{split}
\label{eq: def of kron}
\end{align}
where $\epsilon_{j,i}$ ($j=1,2,3$), $\psi_{i,k}$, $\rho^0_{i,k}$ and $\rho^1_{i,k}$ are induced by the approximation truncation errors defined in (\ref{eq: kron expres}). $\rho^2_{i,k} = \int_{t_k}^{t_{k+1}} \tilde{u}_i^\top R (\hat{u}_i+u_i) \mathrm{d}t$ is also induced by the approximation truncation errors.

\begin{assumption}
Given $N \in \mathbb{N}_+$, there exist $L^* \in \mathbb{N}_+$ and $\alpha>0$, such that for all $L>L^*$ and $i\in \mathbb{N}_+$, the following inequality holds:
\begin{align}
    \frac{1}{L} M_i^\top M_i \geq \alpha I.
\label{eq: assumption for full rank1}
\end{align}
\label{ass: full rank1}
\end{assumption}
\begin{remark}
Assumption \ref{ass: full rank1} is reminiscent of the persistent excitation (PE) condition \cite{Jiang_book2021,astrom1997}. It is needed to guarantee the uniqueness of the least-square solution to \eqref{eq: linear exp Diff V}, and prove the convergence of the proposed data-driven PI algorithm. As in the literature of ADP-based data-driven control \cite{Book_Jiang, Book_Lewis}, one can fulfill it by means of added exploration noise, such as sinusoidal signals and random noise.
\end{remark}%

By \eqref{eq: kron expres} and the definitions of $M_{i,k}$, $Y_{i,k}$ and $E_{i,k}$ in \eqref{eq: def of kron}, \eqref{eq: Diff V with P} is finally transferred as a linear equation with respect to $\Upsilon^N_i$,
\begin{align}\label{eq:kthsegment}
    M_{i,k} \Upsilon^N_i + E_{i,k}  = Y_{i,k}.
\end{align}
Combining equations of \eqref{eq:kthsegment} from $k=1$ to $k=L$, we have
\begin{align}
    M_i \Upsilon^N_i + E_{i}  = Y_i.
\label{eq: linear exp Diff V}
\end{align}
Let $\hat{E}_{i}$ be defined such that 
\begin{align}
     \hat{E}_{i}  = Y_i - M_i \hat{\Upsilon}^N_i.
\label{eq: appro err}
\end{align}

Under Assumption \ref{ass: full rank1}, the method of least squares can be applied to minimize $\hat{E}_i^\top \hat{E}_i$, i.e. $\hat{E}_i^\top \hat{E}_i$ can be minimized by $\hat{\Upsilon}_i$, of which the expression is
\begin{align}
    \hat{\Upsilon}^N_i  = M_i^\dagger Y_i.
\label{eq: lst sq}
\end{align}
With the result of $\hat{\Upsilon}^N_i$ in (\ref{eq: lst sq}),  $\hat{P}_{j,i} (j=0 \cdots 2)$ and $\hat{K}_{j,i}(j=0,1)$ can be reconstructed by (\ref{eq: reconstruct P}) and (\ref{eq: reconstruct K}) respectively.

The detailed data-driven PI algorithm is shown in Algorithm \ref{algo: PI}. As shown from (\ref{eq: def of kron}), $M_i$ and $Y_i$ are constructed by the input-state data sampled along the trajectories of the system. Therefore, no model information is required for the computation of $\hat{\Upsilon}^N_i$. Furthermore, since the trajectories data is collected once and reused throughout the iterations, Algorithm \ref{algo: PI} is called off-policy.

\begin{remark}
Due to the property that $P^\top_{2,i}(\xi, \theta) = P_{2,i}(\theta, \xi)$, the diagonal elements of $P_{2,i}$ satisfy $\text{diag}(P_{2,i}(\xi, \theta)) = \text{diag}(P_{2,i}(\theta, \xi))$. Therefore, the vector of basis functions $\Psi$ must satisfy $\Psi(\xi,\theta) =\Psi(\theta,\xi)$ to approximate such functions.
\end{remark}

\begin{remark}
When calculating $M_i$ and $Y_i$ in Algorithm \ref{algo: PI}, integrals in \eqref{eq: definition of Gm} should be calculated. In practice, we can use Riemann sum to calculate the integrals in \eqref{eq: definition of Gm}, for example midpoint, trapezoid, and Simpson's rules.
\end{remark}
The convergence of the data-driven PI algorithm is studied. The following lemma shows that at each iteration, the value functional and the updated control policy are well approximated, as long as the number of basis functions is large enough.

\begin{lemma}
 Let ${P}_{0,i}, {P}_{1,i}, {P}_{2,i}$ be the solution to (\ref{eq: policy evaluation}), and ${K}_{0,i+1},{K}_{1,i+1}$ be the solution to (\ref{eq: policy improvement}). Under Assumption \ref{ass: full rank1} and given an admissible controller ${u}_{1}(x_t) = -{K}_{0,1}x(t) - \int_{-\tau}^{0}{K}_{1,1}(\theta)x_t(\theta) \mathrm{d}\theta$, for any $i \in \mathbb{N}_+$ and $\eta>0$, there exists a positive integer $N^*>0$, such that if $N>N^*$, the following results hold
\begin{align}
\begin{split}
    &\lvert \hat{P}_{0,i} - {P}_{0,i} \lvert \leq \eta, \quad \lVert \hat{P}_{1,i} - {P}_{1,i} \lVert_\infty \leq \eta,  \\
    &\lVert \hat{P}_{2,i} - {P}_{2,i} \lVert_\infty \leq \eta, \quad \lvert \hat{K}_{0,i+1} - {K}_{0,i+1} \lvert \leq \eta,\\
    &\lVert \hat{K}_{1,i+1} - {K}_{1,i+1} \lVert_\infty \leq \eta.
\end{split}
\end{align}
\label{lm: model-based PI}
\end{lemma}
\begin{proof}
Define $\Tilde{\Upsilon}^N_i = {\Upsilon}^N_i - \hat{\Upsilon}^N_i$. Subtracting (\ref{eq: appro err}) from (\ref{eq: linear exp Diff V}) yields 
\begin{align}
    \hat{E}_{i} = M_i \Tilde{\Upsilon}^N_i + {E}_{i}.
\label{eq: err rewritten}
\end{align}
Since $\hat{E}_i^\top \hat{E}_i$ is minimized by the method of least squares, the following inequality holds
\begin{align}
    \frac{1}{L}\hat{E}_{i}^\top \hat{E}_{i} \leq  \frac{1}{L}{E}_{i}^\top {E}_{i}.
\label{eq: err bound}
\end{align}
Furthermore, according to (\ref{eq: err rewritten}) and (\ref{eq: err bound}), we have 
\begin{align}
\begin{split}
    \frac{1}{L} \Tilde{\Upsilon}^{N\top}_i M_i^\top M_i \Tilde{\Upsilon}^N_i &= \frac{1}{L} (\hat{E}_i - E_i)^\top (\hat{E}_i - E_i) \\
    & \leq \frac{4}{L}E_i^\top E_i.
\end{split}
\end{align}
Therefore, via Assumption \ref{ass: full rank1}, the following inequality holds
\begin{align}
    \Tilde{\Upsilon}_i^{N\top} \Tilde{\Upsilon}^N_i \leq \frac{4}{\alpha L} E_i^\top E_i \leq \frac{4}{\alpha} \max_{1\leq k \leq L} E_{i,k}^2. 
\label{eq: leq Upsilon}
\end{align}
Then the lemma will be proved by induction.
When $i=1$, $\hat{u}_1=u_1$, and therefore $\tilde{u}_1=0$, and $\rho^0_{1,k} = \rho^1_{1,k} = \rho^2_{1,k} = 0$. Furthermore, according to (\ref{eq: trun err}), (\ref{eq: kron expres}), and (\ref{eq: def of kron}), we obtain that for any $1\leq k \leq L$ and $\eta>0$, there exists $N^*>0$, such that if $N>N^*$, $E_{i,k}^2\leq \alpha \eta$. Therefore, by \eqref{eq: leq Upsilon}, when $i=1$ 
\begin{align}
    \lim_{N \rightarrow \infty} \Tilde{\Upsilon}_i^{N\top} \Tilde{\Upsilon}^N_i = 0.
\label{eq: norm Upsilon}
\end{align}
According to (\ref{eq: approximate formula}), (\ref{eq: reconstruct P}), (\ref{eq: reconstruct K}), and the boundedness of the functions $\Phi(\theta)$, $\Psi(\xi, \theta)$, $\Lambda(\xi, \theta)$ on the compact interval $\theta,\xi \in [-\tau, 0]$, the lemma holds for $i=1$.

Suppose for some $i>1$, we have $\lim_{N\rightarrow\infty}\lvert \hat{P}_{0,i-1} - {P}_{0,i-1} \lvert =0$, $\lim_{N\rightarrow\infty}\norm{ \hat{P}_{j,i-1} - {P}_{j,i-1}}_\infty =0(j=1,2)$,  $\lim_{N\rightarrow\infty}\lvert \hat{K}_{0,i} - {K}_{0,i} \lvert =0$, and $\lim_{N\rightarrow\infty}\norm{\hat{K}_{1,i} - {K}_{1,i}}_\infty =0$. Then, by the definitions of $\rho^0_{i,k}$, $\rho^1_{i,k}$, and $\rho^2_{1,k}$, we have $\lim_{N\rightarrow\infty}\lvert 2\rho^0_{i,k} + 2\rho^1_{i,k}+\rho^2_{1,k}\lvert =0$ for any $1 \leq k \leq L$. Combining this with (\ref{eq: trun err}), we have $\lim_{N\rightarrow \infty}\lvert E_{i,k} \lvert = 0$. Hence, by \eqref{eq: leq Upsilon}, (\ref{eq: norm Upsilon}) can be derived for $i>1$. The proof is thus completed. 
\end{proof}

\begin{algorithm}[t]
	\caption{Data-driven Policy Iteration}\label{algo: PI}
	\begin{algorithmic}[1]
	    \State Choose the vector of the basis functions $\Phi(\theta)$, $\Psi(\xi,\theta)$, and $\Lambda(\xi,\theta)$.
	    \State Choose the sampling instance $t_{k} \in [t_1, t_{L+1}]$.
		\State Choose input $u = u_1 + e$, with $e$ an exploration signal, to explore system (\ref{eq: Time delay system}) and collect the input-state data $u(t), x(t), t\in [0, t_{L+1}]$. Set the threshold $\delta>0$ and $i=1$.
		\State \textbf{repeat}
		\State \indent Calculate $\hat{u}_i(t) = \hat{u}_i(x_t)$ along the trajectory of $x$. 
		\State \indent Construct $M_i$ and $Y_i$ by (\ref{eq: def of kron}).
		\State \indent \textbf{while} Assumption \ref{ass: full rank1} is not satisfied
		\State \indent \indent Collect more data and insert it into $M_i$ and $Y_i$.
		\State \indent \textbf{end while}
		\State \indent Get $\hat{\Upsilon}^N_i$ by solving (\ref{eq: lst sq}).
		\State \indent Get $\hat{K}_{0,i+1}$ and $\hat{K}_{1,i+1}$ by (\ref{eq: reconstruct K}). 
	    \State \indent $\hat{u}_{i+1}(x_t) = -\hat{K}_{0,i+1}x(t) - \int_{-\tau}^{0} \hat{K}_{1,i+1}(\theta) x_t(\theta) \mathrm{d}\theta$
		\State \indent $i \leftarrow i+1$
		\State \textbf{until} $\lvert \hat{\Upsilon}^N_{i} - \hat{\Upsilon}^N_{i-1} \lvert < \delta $.
		\State Use $\hat{u}_{i}(x_t)$ as the control input.
	\end{algorithmic}
\end{algorithm}

\begin{theorem}
Given an admissible initial controller $u_1$, for any $\eta>0$, there exist integers $i^*>0$ and $N^{**}>0$, such that
\begin{align}
\begin{split}
    & \lvert \hat{P}_{0,i^*} - {P}_{0}^* \lvert \leq \eta, \quad \lVert \hat{P}_{1,i^*} - {P}_{1}^* \lVert_{\infty} \leq \eta, \\
    & \lVert \hat{P}_{2,i^*} - {P}_{2}^* \lVert_{\infty} \leq \eta, \quad \lvert \hat{K}_{0,i^*+1} - {K}_{0}^* \lvert \leq \eta, \\
    & \lVert \hat{K}_{1,i^*+1} - {K}_{1}^* \lVert_{\infty} \leq \eta.
\end{split}    
\end{align}
\label{thm: PI converge}
if $N>N^{**}$.
\end{theorem}
\begin{proof}
Here we take $\hat{P}_{1,i}$ as an example. According to Theorem \ref{thm: model-based PI}, there exists $i^*>0$, such that 
\begin{align}
    \lVert P_{1,i^*}(\theta) - P_{1}^*(\theta) \lVert_\infty \leq \eta/2.
\end{align}
According to Lemma \ref{lm: model-based PI}, there exists $N^{**}>0$, such that if $N>N^{**}$, 
\begin{align}
    \lVert \hat{P}_{1,i^*}(\theta) - P_{1,i^*}(\theta) \lVert_\infty \leq \eta/2.
\end{align}
Therefore, by the triangle inequality, we have
\begin{align}
\begin{split}
    \lVert \hat{P}_{1,i^*}(\theta) - P_{1}^*(\theta) \lVert_\infty &\leq \lVert \hat{P}_{1,i^*}(\theta) - P_{1,i^*}(\theta) \lVert_\infty \\
    &+ \lVert {P}_{1,i^*}(\theta) - P^*_{1}(\theta) \lVert_\infty \leq \eta.    
\end{split}
\end{align}
\end{proof}

According to Theorem \ref{thm: PI converge}, we see that $\hat{K}_{j,i+1}(j=0,1)$ calculated by Algorithm \ref{algo: PI} converges to ${K}^*_{j}(j=0,1)$ as the iteration step of the algorithm and the number of basis functions tend to infinity. Therefore, the proposed data-driven PI solves Problem \ref{pbm: PI}.

\section{Practical Applications}
In this section, we demonstrate the effectiveness of the proposed data-driven PI algorithms by two practical examples, with regards to regenerative chatter in metal cutting and connected and autonomous vehicles (CAVs) in mixed traffic consisting of both autonomous vehicles (AVs) and human-driven vehicles (HDVs).

\subsection{Regenerative Chatter in Metal Cutting}
Consider the example of regenerative chatter in metal cutting \cite[Example 1.1]{Book_Keqin}, \cite{Mei2005}, where the thrust force of the tool is proportional to the instantaneous chip thickness $([x(t)]_1 - [x(t-\tau)]_1)$, leading to the time-delay effect. Then the model can be described by (\ref{eq: Time delay system}) with $A = \begin{bmatrix}0 &1\\ -(k+F_t/m) & -c/m \end{bmatrix}$, $A_d = \begin{bmatrix}0 &0\\ F_t/m & 0 \end{bmatrix}$ and $B = \begin{bmatrix}0 \\ 1/m\end{bmatrix}$. In this example, The parameters are chosen as $m=2$, $c=0.2$, $k=10$, $F_t=1$, and $\tau=1.3$. The initial state of the system is $[x_0(\theta)]_i = 10\sum_{j=1}^{50}\sin{w_{i,j} \theta} + [\chi]_i$ for $i=1,2$, where $w_{i,j}$ and $[\chi]_i$ are randomly sampled from the uniform distribution over $[-10,10]$. The initial admissible controller is $u_1(x_t) = -K_{0,1}x(t) - \int_{-\tau}^{0} K_{1,1}(\theta) x_t(\theta)\mathrm{d}\theta$, with $K_{0,1} = \begin{bmatrix} 1.7417 & 3.9239\end{bmatrix}$ and $K_{1,1}(\theta) = 0$. The exploration noise is set as $u_d(t) = 20\sum_{i=1}^{50}\sin{\omega_i t}$, where $\omega_i$ is randomly sampled from an independent uniform distribution over $[-10,10]$. $u = u_1 + u_d$ is applied to collect the input-state data from the system. For the performance index (\ref{eq: cost function}), $Q=\text{diag}([100,100])$ and $R = 1$. For the basis functions, $\Phi(\theta) = [1, \theta, \theta^2, \theta^3]^\top$, $\Psi(\xi, \theta) = [1, \xi + \theta, \xi^2+\theta^2, \xi\theta, \xi^3+\theta^3, \xi^2\theta + \xi\theta^2, \xi^3\theta+\xi\theta^3, \xi^2\theta^2, \xi^3\theta^2+\xi^2\theta^3, \xi^3\theta^3]^\top$, and $\Lambda(\xi, \theta) = [1, \theta, \theta^2, \theta^3]^\top \otimes [1, \xi, \xi^2, \xi^3]^\top$. 

For the proposed data-driven PI algorithm, the threshold is set as $\delta = 10^{-3}$. As shown in Fig. \ref{fig: Upsilon}, the weight of the basis functions $\hat{\Upsilon}$ converges after eight iterations. In order to inspect the evolution of the performance index with respect to the iteration, we compare the controllers updated at each iteration for the same initial state $x_0$, of which the result is shown in Fig. \ref{fig: Cost_PI}. It is obvious that the performance index decreases with the iteration of the data-driven PI algorithm. The responses of the state with the initial controller and the learned ADP controller are compared in Fig. \ref{fig: Comparison_PI}. The values of the performance index are $J(x_0, \hat{u}_1) = 5.8941\cdot10^4$ and $J(x_0, \hat{u}_{8}) = 3.0256\cdot10^4$. 

Algorithm 1 is compared with the semi-discretization method \cite{Insperger2002}, which transfers system \eqref{eq: Time delay system} into a discrete-time delay-free system with an augmented state, i.e.
\begin{align}\label{eq:discrete-time LTI}
    x_d(k+1) = \bar{A}x_dx(k) + \bar{B}u(k).
\end{align}
The sampling period of the discretization is set as $\Delta t = 0.1s$ and the dimension of the augmented state is $n_d=n\frac{\tau}{\Delta t} = 26$. Then, with the accurate model matrices $(\bar{A},\bar{B})$, the model-based discrete-time linear quadratic regulator (DLQR) is applied to calculate the optimal controller for system \eqref{eq:discrete-time LTI}. The discrete-time ADP in \cite{Huang2021} is also applied for \eqref{eq:discrete-time LTI} to generate an adaptive optimal controller using the same length trajectory data as Algorithm 1. Then these three obtained controllers are tested on system \eqref{eq: Time delay system}, and the corresponding results are shown in Table \ref{tab:compCost}. We see that the performance index is minimal under Algorithm 1. This explains why the discretization sacrifice the system performance. Ideally, discrete-time ADP can generate a similar controller as the model-based DLQR, and the performance indices should be similar. The large deviation between discrete-time ADP and model-based DLQR is induced by the fact that the persistent excitation condition for discrete-time ADP is not satisfied. This further illustrates that by semi-discretization, the dramatically increased dimension of the augmented state makes the requirements on the sampled data more demanding.

The robustness of Algorithm \ref{algo: PI} to measurement noise is evaluated. The measurement of the state is disturbed by an independent Gaussian noise. That is $z(t) = x(t) + \xi(t)$, where $z(t)$ is the measured state, and $\xi(t) \sim \mathcal{N}(0,0.2)$ is independently and identically distributed Gaussian noise. In Algorithm \ref{algo: PI}, the trajectory data of $z(t)$ instead of $x(t)$ is applied to construct the matrices $M_i$ and $Y_i$. The result is shown in Fig. \ref{fig: Cost_noise}. Using the noisy data, for the same initial state $x_0$, the performance index converges after the $5$th to $J=3.3541\cdot10^4$. Compared with the performance index without noise, we see that under the influence of measurement noise, Algorithm 1 can still find a near optimal solution. 

\begin{figure}[t]
\begin{subfigure}{.49\linewidth}
	\centering
	\includegraphics[width=\linewidth]{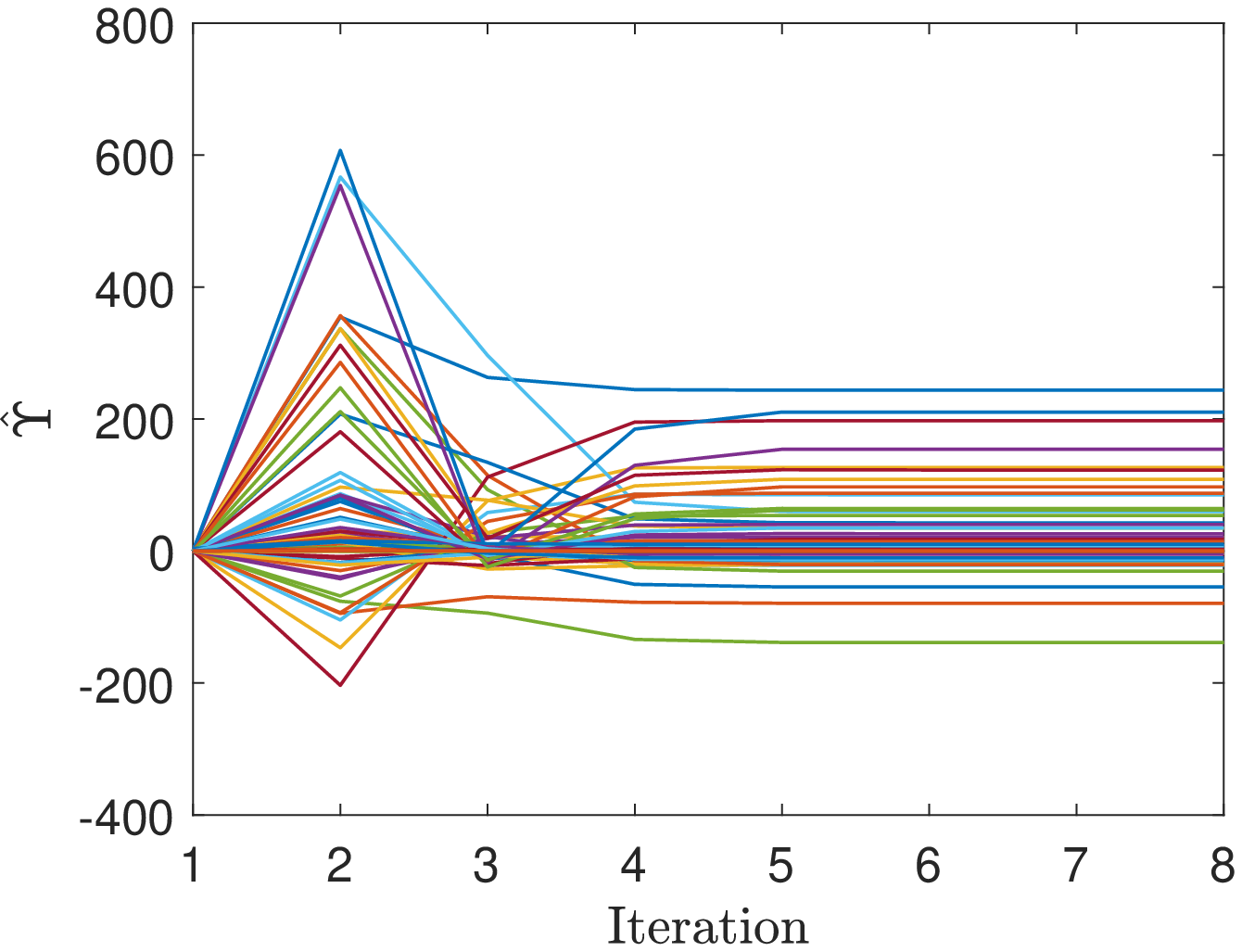}  
	\caption{Evolution of $\hat{\Upsilon}$}
	\label{fig: Upsilon}
\end{subfigure}
\begin{subfigure}{.49\linewidth}
	\centering
	\includegraphics[width=\linewidth]{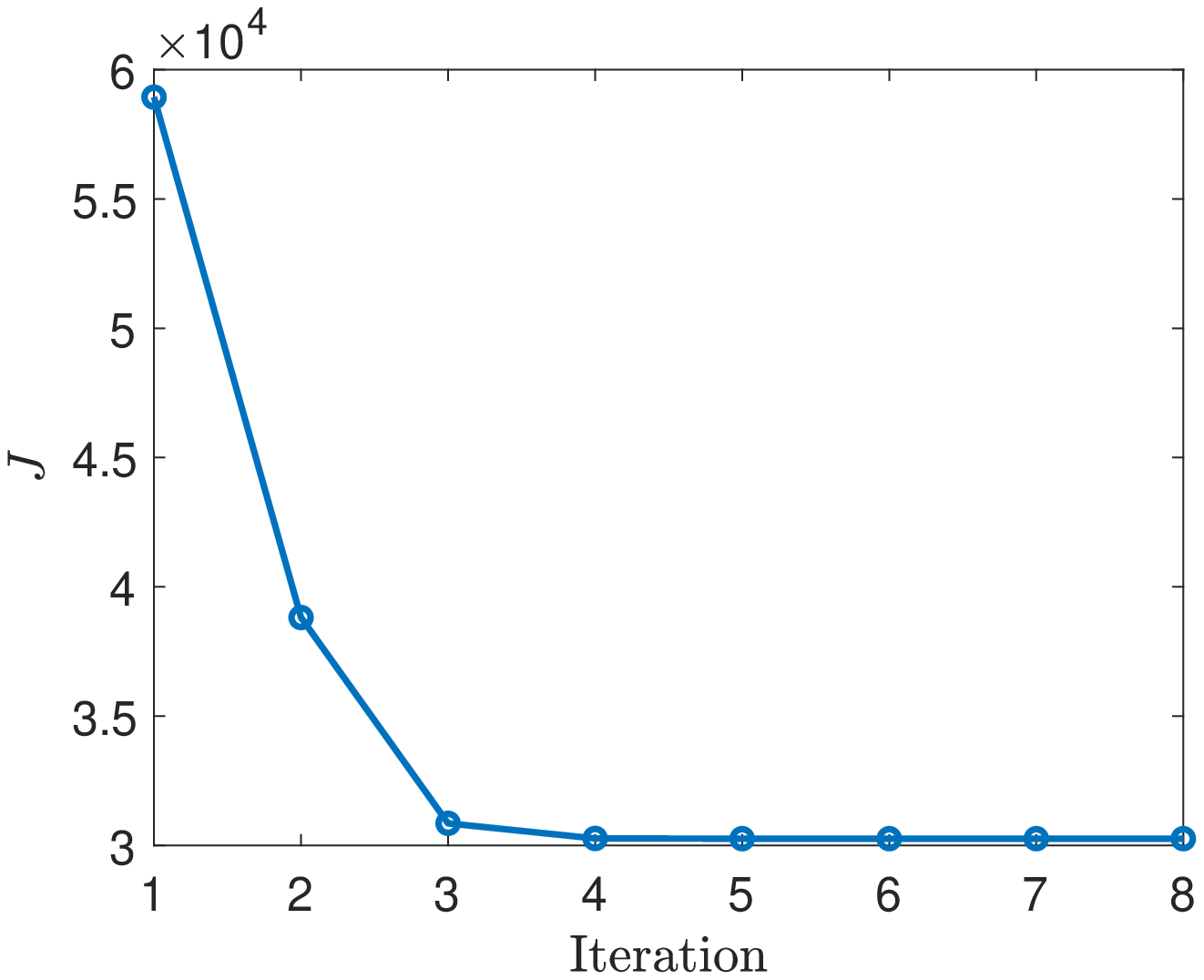}  
	\caption{Evolution of $J$.}
	\label{fig: Cost_PI}
\end{subfigure}
\caption{Evolution of $\hat{\Upsilon}$ and the performance index with respect to iterations.}
\end{figure}  

\begin{figure}[t]
	\centering
	\includegraphics[width=0.75\linewidth]{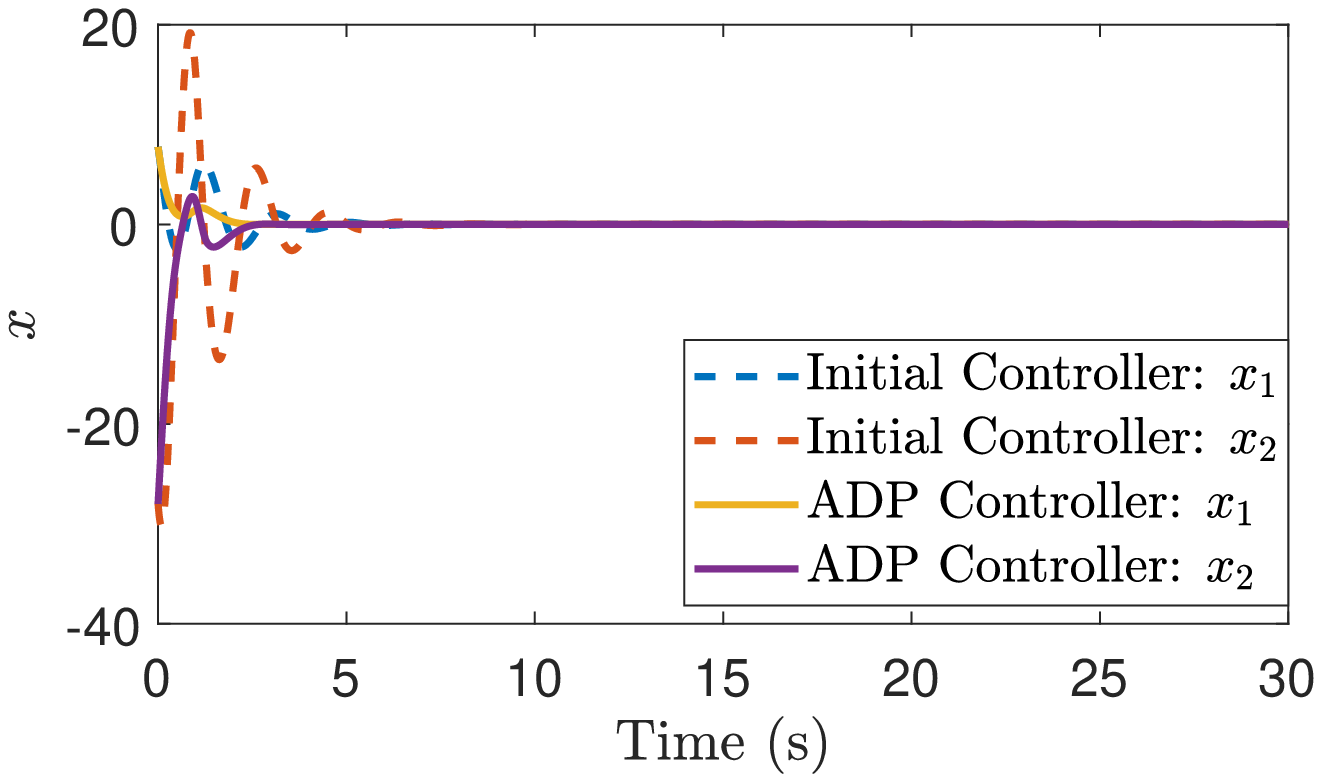}  
	\caption{ADP controller of PI.}
	\label{fig: Comparison_PI}
\caption{Comparison between the initial controller and the ADP controller for regenerative chatter in metal cutting.}
\end{figure} 

\begin{table}[tb]
    \centering
\caption{Comparison of Alg. \ref{algo: PI} and  semi-discretization.}
\begin{tabular}{ |p{2.4cm}|p{2.4cm}|p{2.4cm}|  }
 \hline
 \multicolumn{3}{|c|}{Performance index} \\
 \hline
 Algorithm \ref{algo: PI} & Model-based  DLQR & Discrete-time ADP\\
 \hline
 $3.0256\cdot10^4$ & $3.2614\cdot10^4$   & $4.8203\cdot10^4$   \\
 \hline
\end{tabular}
    \label{tab:compCost}
    \vspace{5mm}
\end{table}

\begin{figure}[t]
	\centering
	\includegraphics[width=0.5\linewidth]{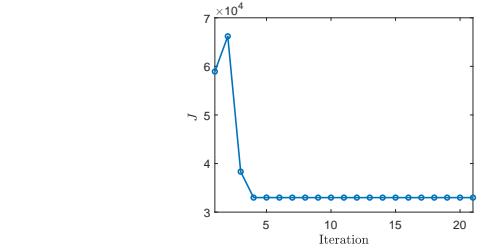}  
	\caption{Evolution of the performance index using noisy data.}
	\label{fig: Cost_noise}
\end{figure} 

\subsection{CAVs in Mixed Traffic}
\begin{figure}[t]
	\centering
	\includegraphics[width=\linewidth]{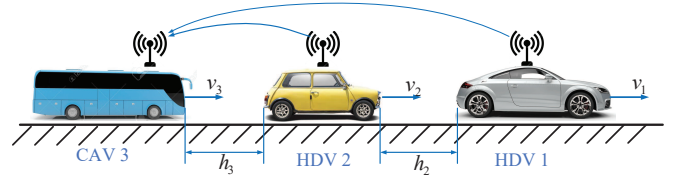}  
	\caption{A string of HDVs and an AV.}
	\label{fig: CAVs}
\end{figure}    

Consider a string of two HDVs and one AV as shown in Fig. \ref{fig: CAVs}, where $h_i$ denotes the bumper-to-bumper distance between the $i$th vehicle and $(i-1)$th vehicle, and $v_i$ denotes the velocity of the $i$th vehicle. Define $\Delta h_i = h_i - h^*$ and $\Delta v_i = v_i - v^*$, where $(h^*, v^*)$ is the equilibrium of the vehicles. $h^*$ depends on the human parameters and $v^* = v_1$. Assuming the velocity of the leading vehicle is constant, and considering the time-delay effect caused by human drivers' reaction time, the system can be described as a linear time-delay system (\ref{eq: Time delay system}) with $x = [\Delta h_2, \Delta v_2, \Delta h_3, \Delta v_3]^\top$, $A = \begin{bmatrix} 0 &-1 &0 &0 \\ 0 &0 & 0 &0\\0 &1 & 0 &-1\\0 &0 & 0 &0 \end{bmatrix}$, $A_d = \begin{bmatrix}0 &0 & 0 &0\\ \alpha_2 
c^* & -(\alpha_2+\beta_2) &0 &0 \\ 0 &0 & 0 &0\\0 &0 & 0 &0\end{bmatrix}$, and $B = \begin{bmatrix} 0 \\0 \\0 \\1\end{bmatrix}$, where $\alpha_2$ and $\beta_2$ denote the human driver parameters and $c^*$ is the derivative of the range policy \cite{Huang2021, Ge2017}. In the simulation, the human parameters are set as $\alpha_2 = 0.1$, $\beta_2 = 0.2$, $\tau = 1.2$, and $c^* = 1.5708$. The weighting matrix of the performance index is $Q = \text{diag}([1,1,10,10])$, and $R=1$. The initial state of the system is $[x_0(\theta)]_i = 30\sum_{j=1}^{10}\sin{w_{i,j} \theta} + [\chi]_i$ for $i=1,2,3,4$, where $w_{i,j}$ and $[\chi]_i$ are randomly sampled from the uniform distribution over $[-10,10]$ and $[-30,30]$ respectively. The initial admissible controller is $u_1(x_t) = -K_{0,1}x(t) - \int_{-\tau}^{0} K_{1,1}(\theta) x_t(\theta)\mathrm{d}\theta$, with $K_{0,1} = \begin{bmatrix} -0.0897 & -0.2772 &-0.3 &0.5196\end{bmatrix}$ and $K_{1,1}(\theta) = 0$. The exploration noise is set as $u_d(t) = \sum_{i=1}^{200}\sin{\omega_i t}$, where $\omega_i$ is randomly sampled from a independent uniform distribution over $[-100,100]$. $u = u_1 + u_d$ is applied to collect the input-state data from the system. The basis functions are  same as these in the example of regenerative chatter in metal cutting. The analytical expressions of the optimal values $K^*_0$ and $K^*_1$ are derived by the method in \cite{Ge2017}, where the precise model is assumed known.

\begin{figure}[t]
\begin{subfigure}{.49\linewidth}
	\centering
	\includegraphics[width=\linewidth]{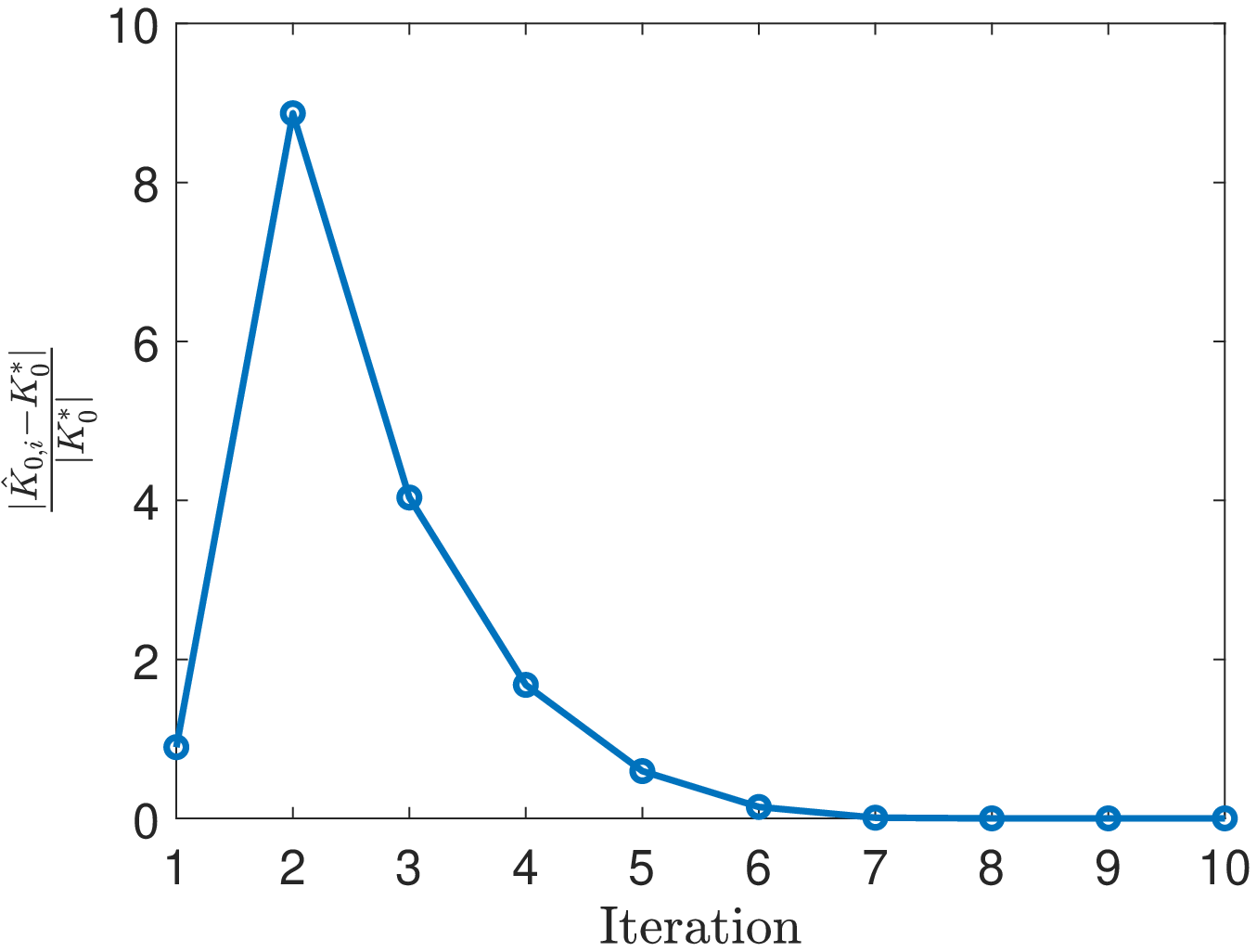}  
	\label{fig: K0_norm_PI}
\end{subfigure}
\begin{subfigure}{.49\linewidth}
	\centering
	\includegraphics[width=\linewidth]{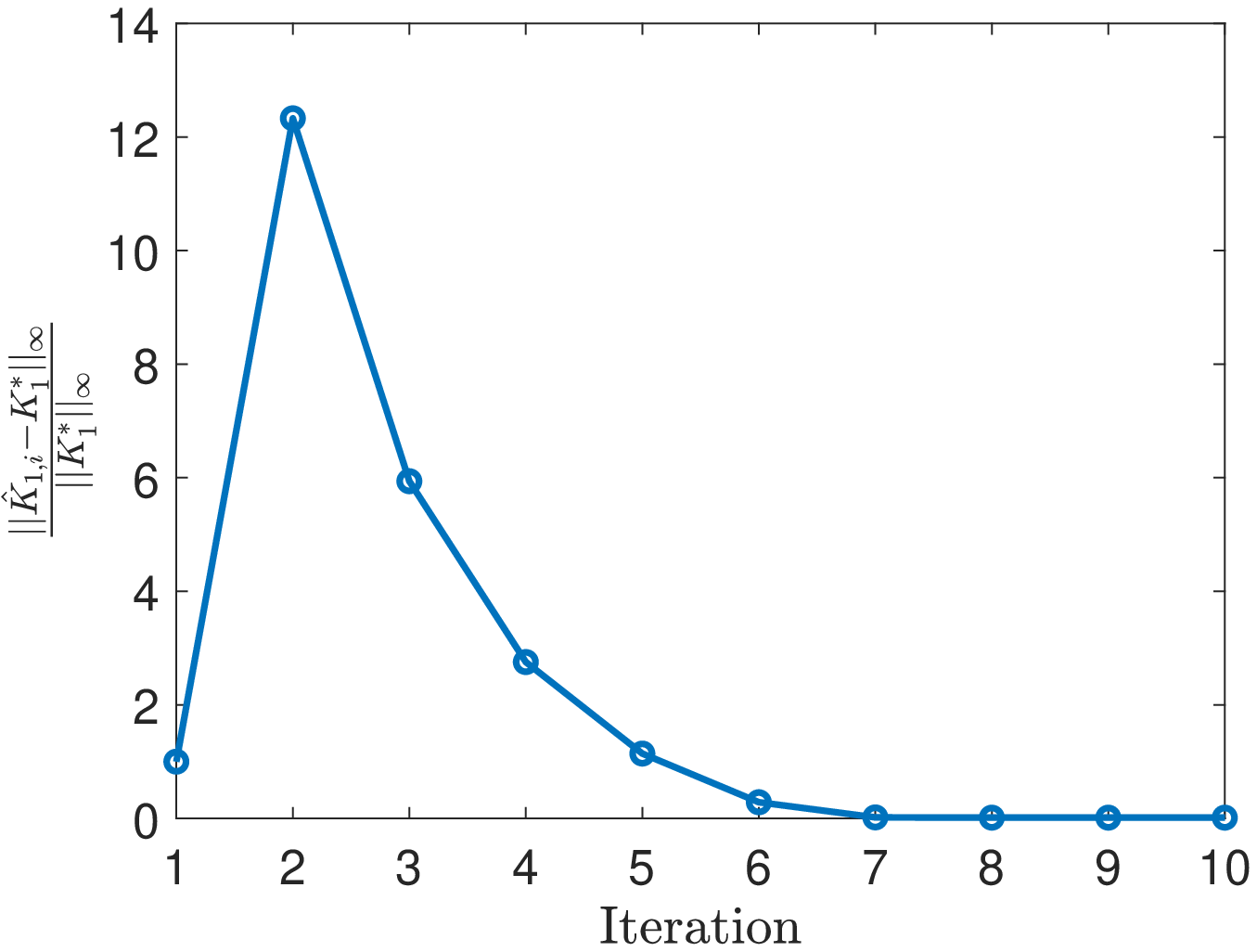}  
	\label{fig: K1_norm_PI}
\end{subfigure}
\caption{Convergence of $\hat{K}_{0,i}$ and $\hat{K}_{1,i}(\theta)$ to the optimal values ${K}^*_{0}$ and ${K}^*_{1}(\theta)$ by PI algorithm. }
\label{fig: Evolution of PI}
\end{figure}  

\begin{figure}[t]
	\centering
	\includegraphics[width=.7\linewidth]{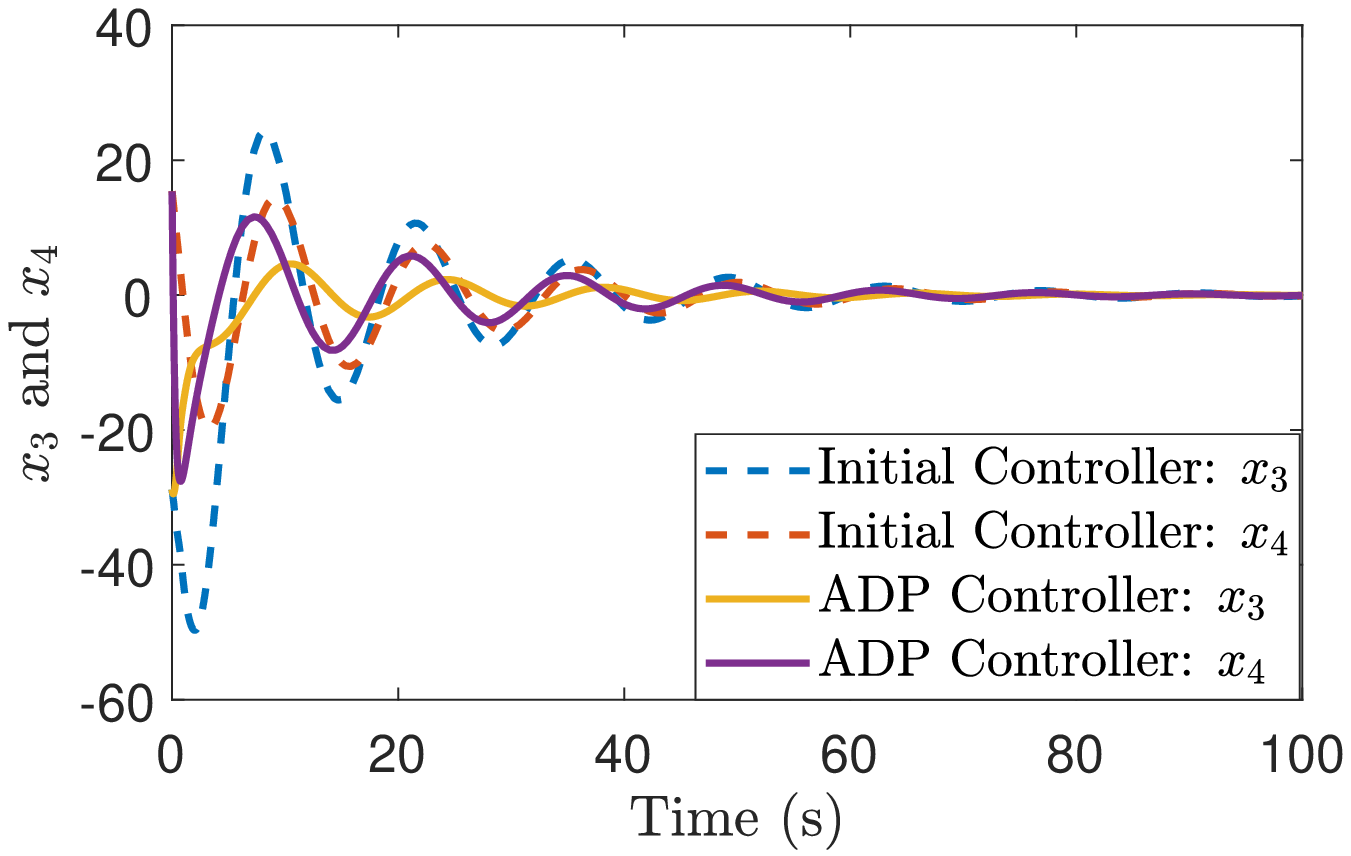}  
	\caption{ADP controller of PI}
	\label{fig: Comparision_4dim_PI}
\caption{Comparison between the initial controller and the ADP controller for CAVs.}
\end{figure}

For our proposed data-driven PI algorithm, the aforementioned initial admissible control $u_1$ is applied to start the algorithm. The threshold is $\delta=10^{-3}$. The convergence of $\hat{K}_{0}$ and $\hat{K}_{1,i}$ are shown in Fig. \ref{fig: Evolution of PI}. At the last iteration of PI, $\frac{|\hat{K}_{0,10}-K^*_{0}|}{|K^*_{0}|} = 0.0008$ and $\frac{||\hat{K}_{1,10}-K^*_{1}||_\infty}{||K^*_{1}||_\infty} = 0.0137$. Therefore, the proposed data-driven PI algorithm can find a sub-optimal solution. The performance comparisons of the initial controller and the ADP controller are shown in Fig. \ref{fig: Comparision_4dim_PI}. Because $x_1$ and $x_2$ are the states of the HDV2, which cannot be influenced by the controller for the AV, they are not plotted in the figure. From the figure, it is seen that with the ADP controller, $x_3$ converges to the equilibrium more quickly than the initial controller. For the values of performance index, $J(x_0,u_1) = 1.4627\cdot 10^{5}$ and  $J(x_0,u_{10}) = 4.7288\cdot 10^{4}$. Therefore, the proposed data-driven PI algorithm can minimize the performance index and improve the performance of the closed-loop system as a consequence.


\section{Conclusions}
This paper has proposed for the first time a novel data-driven PI algorithm for a class of linear time-delay systems described by DDEs. The first major contribution of this paper is to generalize the well-known Kleinman algorithm \cite{Kleinman1968} -- a model-based PI algorithm -- from linear time-invariant systems to linear time-delay systems. The second major contribution of this paper is that we have combined the proposed model-based PI algorithm and RL techniques
to develop a data-driven PI algorithm for solving the direct adaptive optimal control problem for linear time-delay systems with completely unknown dynamics. The efficacy of the proposed learning-based adaptive optimal control design methods has been validated by means of two real-world applications arising from metal cutting and connected vehicles. Our future work will be directed at extending the proposed learning-based control methodology to other practically important classes of time-delay systems such as nonlinear systems and multi-agent systems.

\ifCLASSOPTIONcaptionsoff
  \newpage
\fi



%
%
%
\bibliographystyle{ieeetr}        
\bibliography{IEEEabrv,mybibfile}

\end{document}